\documentclass[12pt]{article}

\usepackage[utf8]{inputenc}

\usepackage{bbm}
\usepackage{amsmath,amsfonts,amssymb,amsthm}
\usepackage{xspace}
\usepackage{enumerate}
\usepackage{url}
\usepackage[colorlinks=true,linkcolor=blue,citecolor=red]{hyperref}

\usepackage{tikz}

\usetikzlibrary{
  calc,  % for ($<coordinate calculation>$)
  cd,     % for tikzcd environment
  arrows.meta,
}
\colorlet{pred}{red!75}
\colorlet{pgreen}{green!35}
\tikzset{
  node split radius/.initial=1,
  node split color 1/.initial=pred,
  node split color 2/.initial=blue,
  node split color 3/.initial=pgreen,
  node split half/.style={node split={#1,#1+180}},
  node split/.style args={#1,#2}{
    path picture={
      \tikzset{
        x=($(path picture bounding box.east)-(path picture bounding box.center)$),
        y=($(path picture bounding box.north)-(path picture bounding box.center)$),
        radius=\pgfkeysvalueof{/tikz/node split radius}}
      \foreach \ang[count=\iAng, remember=\ang as \prevAng (initially #1)] in {#2,360+#1}
        \fill[line join=round, draw, fill=\pgfkeysvalueof{/tikz/node split color \iAng}]
          (path picture bounding box.center)
          --++(\prevAng:\pgfkeysvalueof{/tikz/node split radius})
          arc[start angle=\prevAng, end angle=\ang] --cycle;
} } }

\theoremstyle{plain}
\newtheorem{theorem}{Theorem}%[section]
\newtheorem{lemma}[theorem]{Lemma}
\newtheorem{proposition}[theorem]{Proposition}
\newtheorem{corollary}[theorem]{Corollary}
\theoremstyle{definition} % This does not use italics for the text.
\newtheorem{remark}[theorem]{Remark}

\newcounter{claim}
\renewcommand{\theclaim}{\Alph{claim}}
\newenvironment{claim}{\refstepcounter{claim}%
\par\medskip\par\noindent{\it Claim~\theclaim.~}~\rm}%
{\par\medskip\par}
\newenvironment{subproof}{\par\medskip\par\noindent{\sl Proof of Claim~\theclaim.~}}%
{$\,\triangleleft$\par\bigskip\par}
\newenvironment{subproofof}[1]{\par\noindent{\sl Proof of Claim~\ref{cl:#1}.~}}%
{$\,\triangleleft$\par\bigskip\par}

\makeatletter
\def\@gifnextchar#1#2#3{\let\@tempe#1\def\@tempa{#2}\def\@tempb{#3}%
  \futurelet\@tempc\@gifnch}

\def\@gifnch{\ifx\@tempc\@sptoken\let\@tempd\@tempb%
  \else\ifx\@tempc\@tempe\let\@tempd\@tempa\else\let\@tempd\@tempb\fi\fi\@tempd}

\def\SK@set#1{\left\{#1\right\}}
\def\SK@@set#1#2{\{#1\,:\,
    \begin{array}{@{}l@{}}#2\end{array}
\}}
\def\SK@mset#1{\left\{\!\!\left\{#1\right\}\!\!\right\}}
\def\SK@@mset#1#2{\{\!\!\{#1\,:\,
    \begin{array}{@{}l@{}}#2\end{array}
\}\!\!\}}
\def\BIG@set#1{\Big\{#1\Big\}}
\def\BIG@@set#1#2{\Big\{#1\:\Big|\:
    \begin{array}{@{}l@{}}#2\end{array}
\Big\}}
\newcommand{\Set}[1]{\@gifnextchar\bgroup{\SK@@set{#1}}{\SK@set{#1}}}
\newcommand{\Mset}[1]{\@gifnextchar\bgroup{\SK@@mset{#1}}{\SK@mset{#1}}}
\newcommand{\Bigset}[1]{\@gifnextchar\bgroup{\BIG@@set{#1}}{\BIG@set{#1}}}
\makeatother

\newcommand{\refeq}[1]{(\ref{eq:#1})}

\newcommand{\of}[1]{\left( #1 \right)}

\newcommand{\function}[2]{:#1 \rightarrow #2}

\newcommand{\bZ}{\mathbb{Z}}
\newcommand{\cP}{\mathcal{P}}

\newcommand{\prob}[1]{\mathsf{P}[ #1 ]}
\newcommand{\probb}[1]{\mathsf{P}\left[ #1 \right]}
\newcommand{\cprob}[2]{\mathsf{P}[ #1 \mid #2 ]}
\newcommand{\bind}{\mathrm{Bin}}
\newcommand{\mean}{\mathsf{E}}
\DeclareMathOperator{\var}{\mathsf{Var}}
\DeclareMathOperator{\cov}{\mathsf{cov}}

\newcommand{\one}{\mathbbm{1}}
\newcommand{\cO}{\mathcal{O}}

\title{Canonization of a random graph by two matrix-vector multiplications\footnote{A preliminary
    version of this paper was presented at the 31st Annual European Symposium on Algorithms,
    ESA 2023~\cite{ESA23}.}}

\author{Oleg Verbitsky\thanks{Institut f\"{u}r Informatik, Humboldt-Universit\"{a}t zu Berlin, Germany.
    On leave from the IAPMM, Lviv, Ukraine.}
  \and
Maksim Zhukovskii\thanks{School of Computer Science, University of Sheffield, UK.}}

\date{}

\begin{document} 

\maketitle

\begin{abstract}
  We show that a canonical labeling of a random graph on $n$ vertices can be
  obtained by assigning to each vertex $x$ the triple $(w_1(x),w_2(x),w_3(x))$,
  where $w_k(x)$ is the number of walks of length $k$ starting from $x$.
  This takes time $\cO(n^2)$, where $n^2$ is the input size,
  by   using just two matrix-vector multiplications. The linear-time canonization
  of a random graph is the classical result of Babai, Erd\H{o}s, and Selkow.
  For this purpose they use the well-known combinatorial color refinement
  procedure, and we make a comparative analysis of the two algorithmic approaches.
\end{abstract}

\bigskip\bigskip


\bigskip

\section{Introduction}\label{s:intro}

A \emph{walk} in a graph $G=(V,E)$ is a sequence of vertices $x_0x_1\ldots x_k$ such that
$(x_i,x_{i+1})\in E$ for every $0\le i<k$. We say that $x_0x_1\ldots x_k$
is a walk of \emph{length $k$ from $x_0$ to $x_k$}.
For a vertex $x\in V$, let $w^G_k(x)$ denote the total number of walks of length $k$ in $G$ starting from $x$.
Furthermore, we define $\mathbf{w}^G_k(x)=(w^G_1(x),\ldots,w^G_k(x))$.

The Erd\H{o}s-R\'enyi random graph $G(n,p)$ is a graph on the vertex set $[n]=\{1,\ldots,n\}$
where each pair of distinct vertices $x$ and $y$ is adjacent with probability $p$ independently
of the other pairs. In particular, $G(n,1/2)$ is a random graph chosen equiprobably from among
all graphs on~$[n]$.
Our main result in this paper shows that every vertex $x$ in $G(n,1/2)$ is with high probability
individualized by the numbers of walks from $x$ of lengths 1, 2, and 3, and that the maximum length 3
is optimal for this purpose.

\begin{theorem}\label{thm:3paths}
  Let $G=G(n,1/2)$. Then
  \begin{enumerate}[\bf 1.]
  \item
  $\mathbf{w}^G_3(x)\ne\mathbf{w}^G_3(y)$ for all $x\ne y$
  with probability at least $1-O(n^{-1/2})$.
  \item
  % $\mathbf{w}^G_2(x)\ne\mathbf{w}^G_2(y)$ for some $x\ne y$
  % with probability at least $1-o(1)$.
  $\mathbf{w}^G_2(x)\ne\mathbf{w}^G_2(y)$ for all $x\ne y$
  with probability approaching 0 as $n\to\infty$. 
  \end{enumerate}

\end{theorem}

If $\alpha$ is an isomorphism from a graph $G$ to a graph $H$,
then clearly $\mathbf{w}^G_k(x)=\mathbf{w}^H_k(\alpha(x))$.
Theorem \ref{thm:3paths}, therefore, shows that the map $x\mapsto\mathbf{w}^G_3(x)$
is a canonical labeling of $G$ for almost all $n$-vertex graphs $G$.
This labeling is easy to compute. Indeed,
if $A$ is the adjacency matrix of $G$ and $\one$ is the all-ones vector-column of length $n$, then
$$
(w^G_k(1),\ldots,w^G_k(n))^\top=A^k\one.
$$
After noting that $w^G_1(x)=d(x)$, where $d(x)$ denotes the degree of a vertex $x$,
this yields the following simple canonical labeling algorithm.

\medskip

\noindent\textsc{Algorithm A}

\smallskip

\noindent\textsc{Input:} a graph $G$ on $[n]$ with adjacency matrix $A$.
\begin{enumerate}
\item
  Form a vector $D_1=(d(1),\ldots,d(n))^\top$.
\item
  Compute $D_2=AD_1$ and $D_3=AD_2$.
\item 
  Let $W$ be the matrix formed by the three columns $D_1,D_2,D_3$
  and let $W_1,\ldots,W_n$ be the rows of $W$.
\item
  If there are identical rows $W_x=W_y$ for some $x\ne y$, then give up. Otherwise, 
\item
  to each vertex $x$, assign the label $W_x$.
\end{enumerate}

\begin{corollary}\label{cor:}
  Algorithm A with high probability canonizes a random $n$-vertex graph, taking
  time $\cO(n^2)$ on every input.
\end{corollary}

The notation $\cO(\cdot)$ in the time bound means a linear function up to
the logarithmic factor $\log n\log\log n$ corresponding to the complexity
of integer multiplication \cite{HarveyH21}, that is, $\cO(n^2)=O(n^2\log n\log\log n)$.
If the model of computation assumes that multiplication of two integers
takes a constant time, then we just set $\cO(n^2)=O(n^2)$.
This time bound stems from the fact that
the two matrix-vector multiplications in Step 2 are the most expensive
operations performed by the algorithm. Note that this bound is essentially linear
because a random graph is with high probability dense, having $(1/4+o(1))n^2$ edges.

The linear-time canonization of almost all graphs is a classical result of
Babai, Erd\H{o}s, and Selkow \cite{BabaiES80}, which was a basis for settling
the average-case complexity of graph isomorphism in \cite{BabaiK79}.
While our algorithm is based solely on basic linear-algebraic primitives,
the method used in \cite{BabaiES80,BabaiK79} is purely combinatorial.
Before comparing the two approaches, we put Theorem \ref{thm:3paths}
in the context of the earlier work on walk counts and their applications
to isomorphism testing.

Of course, Algorithm A can be enhanced by taking into account
also longer walks, that is, by involving also other vector-columns $A^k\one$ for $k>3$.
Note that there is no gain in considering these vectors for $k\ge n$.
Indeed, if $A^k\one$ is a linear combination of the vectors $\one, A\one, A^2\one,\ldots, A^{k-1}\one$,
the same is obviously true also for $A^{k+1}\one$ (cf.\ \cite[Lemma 1]{PowersS82}
and see also \cite{Hagos02} for a more detailed linear-algebraic analysis). Therefore, it suffices to
start our matrix $W$ from the column $\one$ and add a subsequent column $A^k\one$
as long as this increases the rank of $W$, which is possible only up to $k=n-1$.
The $n\times n$ matrix $W=W^G$ formed by the columns $\one, A\one, \ldots, A^{n-1}\one$
is called the \emph{walk matrix} of the graph $G$ (\emph{WM} for brevity).
The entries of $W^G=(w_{x,k})_{1\le x\le n,\,0\le k<n}$
are nothing else as the walk counts $w_{x,k}=w^G_k(x)$. Note that $w^G_0(x)=1$ as there is a single
walk of length $0$ from~$x$.

We say that a graph $G$ is \emph{WM-discrete} if all rows of the walk matrix $W^G$ are pairwise different, i.e.,
$W^G_x\ne W^G_y$ for all $x\ne y$, where $W^G_x=(w_0^G(x),w_1^G(x),\ldots,w_{n-1}^G(x))$.
If $G$ is WM-discrete, then the walk matrix yields a canonical labeling
by assigning each vertex $x$ the vector $W^G_x$.
A useful observation is that if $W^G$ has identical rows, then this matrix is singular; cf.~\cite[Section 7]{Godsil12}.
O'Rourke and Touri \cite{ORourkeT16} prove that the walk matrix of a random graph 
is non-singular with high probability. This implies that a random graph is WM-discrete
with high probability and, hence, almost all graphs
are canonizable by computing the $n\times n$ walk matrix similarly to
Algorithm A. We remark that this takes time $O(n^3)$,
which is outperformed by our Corollary \ref{cor:} due to using the truncated
variant of WM of size~$n\times 4$.

Remarkably, non-singular walk matrices can be used to test isomorphism
of two given graphs directly rather than by computing their canonical forms.
If graphs $G$ and $H$ are
isomorphic, then their walk matrices $W^G$ and $W^H$ can be obtained from one another by rearranging the rows.
If the last condition is satisfied, we write $W^G\approx W^H$. This relation between
matrices is efficiently checkable just by sorting the rows in the lexicographic order.
We say that a graph $G$ is \emph{WM-identifiable} if, conversely,
for all $H$ we have $G\cong H$ whenever $W^G\approx W^H$.
Liu and Siemons \cite{LiuS22} prove that if the walk matrix of a graph
is non-singular, then it uniquely determines the adjacency matrix.
This implies by \cite{ORourkeT16} that a random graph is WM-identifiable with high probability.

Note that, by a simple counting argument, almost all $n$-vertex graphs \emph{cannot} be identified
by the shorter version of the walk matrix of size $n\times k$ as long as $k=o(\sqrt{n/\log n})$.
In particular, Theorem \ref{thm:3paths} cannot be extended to the identifiability concept.

The combinatorial approach of Babai, Erd\H{o}s, and Selkow \cite{BabaiES80} is based on the \emph{color refinement} procedure
(\emph{CR} for brevity) dating back to the sixties (e.g.,~\cite{Morgan65}).
CR begins with a uniform coloring of all vertices in an input graph
and iteratively refines a current coloring according to the following principle:
If two vertices are equally colored but have
distinct color frequencies in their neighborhoods, then they get distinct colors in the next refinement
step. The refinement steps are executed as long as the refinement is proper. As soon as
the color classes stay the same, CR terminates and outputs the current coloring
(a detailed description of the algorithm is given in Section \ref{ss:CR}).
CR \emph{distinguishes} graphs $G$ and $H$ if their color palettes are distinct.
A graph $G$ is called \emph{CR-identifiable} if it is distinguishable by CR from every non-isomorphic $H$.
CR can also be used for computing a canonical labeling of a single input graph.
We say that a graph $G$ is \emph{CR-discrete} if CR assigns a unique color to each vertex of $G$.
It is easy to prove that every CR-discrete graph is CR-identifiable.
We do not know whether or not this is true also for the corresponding WM concepts.

Powers and Sulaiman \cite{PowersS82} discuss examples when the CR-partition and the WM-partition are different,
that is, CR and the WM-based vertex-classification algorithm give different results. In particular,
\cite[Fig.~3]{PowersS82} shows a graph which is, in our terminology, CR-discrete but not WM-discrete.
We give a finer information about the relationship between the two algorithmic approaches.

\begin{theorem}\label{thm:WM-vs-CR}\hfill
  \begin{enumerate}[\bf 1.]
  \item
    Every WM-discrete graph is also CR-discrete.
  \item
    Every WM-identifiable graph is also CR-identifiable.
  \item
    There is a graph that is
    \begin{enumerate}[\bf (a)]
    \item
      CR-discrete (hence also CR-identifiable) and
    \item
      neither WM-discrete
    \item
      nor WM-identifiable.
    \end{enumerate}
  \end{enumerate}
\end{theorem}

Theorem \ref{thm:WM-vs-CR} shows that color refinement is at least as good as the WM approach
with regard to canonization of a single input graph as well as testing isomorphism of two
input graphs. Moreover, CR is sometimes more successful with respect to both algorithmic problems.
Thus, WM can be regarded as a weaker algorithmic tool for canonical labeling and isomorphism testing,
which is not so surprising as this approach is actually based on a single basic linear-algebraic primitive,
namely matrix-vector multiplication.
% In this sense, Algorithm A
% is arguably simpler than the classical CR-based canonization of a random graph
% as it
Corollary \ref{cor:} demonstrates, therefore, that a random graph can be canonized in an essentially
linear time even with less powerful computational means than color refinement.

\paragraph{Implications for anonymous models of distributed computation.}
Both CR and WM admit natural interpretations as algorithms executed by anonymous processors in a distributed network. A processor is represented by a node of the network graph. In each communication round, every node broadcasts a message to all of its neighbours, receives the multiset of messages sent by its neighbours, and updates its local state accordingly. This is the \emph{multiset-broadcast} model MB, which occupies one of the lowest levels in the hierarchy of weak distributed-computing models introduced by Hella et al.~\cite{Hella}.

The standard LOCAL model is substantially stronger: in particular, nodes have unique identifiers and can send distinct messages to different neighbours. Nevertheless, on the random graph $G(n,1/2)$, the two models are essentially equivalent: every task solvable with high probability in the LOCAL model can be solved in the MB model also with high probability. This follows from the result of Babai, Erd\H{o}s, and Selkow \cite{BabaiES80} as their canonical labeling algorithm translates into an MB algorithm that, with high probability, assigns unique identifiers to all nodes after two rounds of communication.

Considerable attention in distributed computing has been devoted to the congested version of LOCAL, in which messages are restricted to $O(\log n)$ bits per round; see, e.g., \cite{HirvonenS}. The algorithm of Babai, Erd\H{o}s, and Selkow typically requires messages of length $\Theta(\log^2 n)$ and therefore does not directly apply in the congested setting. By contrast, our WM-based algorithm computes unique identifiers in the congested MB model in three communication rounds, with high probability.  For recent related work, see~\cite{RVZ}.

\paragraph{Outline of proofs.}
Theorems \ref{thm:3paths} and \ref{thm:WM-vs-CR} are proved in Sections \ref{s:3paths}
and \ref{s:WM-vs-CR} respectively. We conclude the current section with giving a high level outline of the proofs.

For Part 1 of Theorem~\ref{thm:3paths}, it suffices to fix an arbitrary pair of vertices, say, 1 and 2, and to show that $\mathbf{w}_3^G(1)=\mathbf{w}_3^G(2)$ with probability $O\left(n^{-5/2}\right)$. While the asymptotics of $\prob{\mathbf{w}_2^G(1)=\mathbf{w}_2^G(2)}$ can be determined exactly by using the local de Moivre--Laplace limit theorem and exponential tail bounds for binomial distributions
(cf.~Equality \eqref{eq:Xij}), including the walks of lengths 3 in the analysis is more challenging. Nevertheless, by conditioning on the 2-neighbourhoods of vertices 1 and 2, deriving the desired upper bound on $\prob{\mathbf{w}_3^G(1)=\mathbf{w}_3^G(2)}$ reduces to the classical Littlewood--Offord problem.

%The main idea that allows us to address this challenge is that after exposing the 2-neighbourhoods of the vertices 1 and 2 under the condition $\mathbf{w}_2^G(1)=\mathbf{w}_2^G(2)$, with high probability there is a rather large set $S$ of pairs of vertices $i,j$ adjacent neither to 1 nor to 2 such that all edges $\{i,j\}\in S$ contribute to the difference $w_3^G(1)-w_3^G(2)$ with exactly the same weight $\Delta(i,j)$ (see the details in Section~\ref{proof:1_1}). It follows that if we additionally expose the adjacencies for all pairs of vertices outside $S$, then $\mathbf{w}_3^G(1)=\mathbf{w}_3^G(2)$ exactly when the total number of edges within $S$ is equal to some fixed number uniquely determined by the exposed edges. The last probability is easily estimated by standard tools.

The proof of Part 2 of Theorem~\ref{thm:3paths} is the most involved and conceptually novel.
We define the random variable $X$ as the total number of
vertex pairs $i,j$ such that $\mathbf{w}_2^G(i)=\mathbf{w}_2^G(j)$. The expected number of such pairs
increases together with $n$, as the aforementioned tight bounds for $\prob{\mathbf{w}_2^G(1)=\mathbf{w}_2^G(2)}$
imply that $\mean X=\Theta(\sqrt{n})$. This is still not enough to conclude that $X\ge1$ with
high probability. The main part of the proof in Section~\ref{proof:1_2} is devoted to showing that
$\var(X)=o((\mean X)^2)$, which implies the desired assertion by Chebyshev's inequality.
Note that for disjoint vertex pairs $i,j$ and $i',j'$, the events $\mathbf{w}_2^G(i)=\mathbf{w}_2^G(j)$
and $\mathbf{w}_2^G(i')=\mathbf{w}_2^G(j')$ are not independent. Thus, they contribute non-trivially
to the variance, making the estimation of $\var(X)$ difficult. The main ingredient of the proof is
Claim~\ref{cl:0} saying that this contribution in $o(n^{-3})$. This estimate follows from the asymptotic
independence of the events $\mathbf{w}_2^G(i)=\mathbf{w}_2^G(j)$ and $\mathbf{w}_2^G(i')=\mathbf{w}_2^G(j')$
subject to the typical 1-neighbourhood of $\{i,j,i',j'\}$ and the equalities $w_1^G(i)=w_1^G(j)$ and
$w_1^G(i')=w_1^G(j')$. This asymptotic independence is quite unexpected since the sets of edges adjacent
to the 1-neighbourhoods of $\{i,j\}$ and $\{i',j'\}$ have a large intersection and their impacts on the
two events are, therefore, far from being independent. The proof of this surprising fact reduces to
establishing a local limit theorem for certain random vectors $\Psi_n$ represented as sums of independent
but not identically distributed random vectors. To this end, we first derive an integral limit theorem
based on a helpful observation that the non-diagonal elements of the covariance matrices of $\Psi_n$
are much smaller than their diagonal elements. The final step is made by using the general framework
for conversion of an integral limit theorem into a local limit theorem provided in~\cite{Mukhin}.

The first two parts of Theorem \ref{thm:WM-vs-CR} follow from the fact that the walk numbers
$w^G_k(x)$ are determined by the color assigned to the vertex $x$ by the CR algorithm; cf.~\cite{Dvorak10,PowersS82}.
The similar relationship holds true for the number $w^G_k(x,y)$ of $k$-walks between two
vertices $x$ and $y$ and the 2-dimensional Weisfeiler-Leman algorithm, which can be seen
as a natural extension of CR for computing a canonical coloring of vertex pairs.
In the special case of strongly regular graphs, this relationship just means that
the value of $w^G_k(x,y)$ is the same for all pairs of adjacent $x$ and $y$ as well as
for all pairs of non-adjacent $x$ and $y$. This property is useful for proving Part 3
of Theorem \ref{thm:WM-vs-CR} as it suggests searching for two partial colorings of
the same strongly regular graph that are distinguished by CR but are indistinguishable
in terms of the walk counts $w^G_k(x,y)$. We find such two colorings for the Shrikhande graph.
The rest of the proof consists in removing the auxiliary vertex colors by gluing two colored versions
of the Shrikhande graph together via a simple connector-gadget so that the colors are simulated
by the changed vertex degrees. We remark that this idea proves to be useful also
in another context. In the subsequent work \cite{ArvindFKV24}, a similar approach is used to
separate the power of the Weisfeiler-Leman algorithm and some spectral characteristics of graphs.

\section{Canonization of a random graph}\label{s:3paths}

\subsection{Probability preliminaries}

Let $X$ be a binomial random variable with parameters $n$ and $p$,
that is, $X=\sum_{i=1}^nX_i$ where $X_i$'s are mutually independent and, for each $i$, we have $X_i=1$ with probability 1 and
$X_i=0$ with probability $1-p$. We use the notation $X\sim\bind(n,p)$ when $X$ has this distribution. 
As well known, $X$ is well-concentrated around its expectation~$np$.

\begin{lemma}[Chernoff's bound; see, e.g.,~{\cite[Theorems A.1.11 and A.1.13]{AlonS16}}]\label{lem:chernoff}
If $X\sim\mathrm{Bin}(n,p)$, then
$$
 \prob{|X-np|>t}\le2e^{-t^2/(3pn)}
$$
for every $0\le t\le np/3$.
\end{lemma}

\begin{lemma}\label{lem:binomial}
  If $X$ and $Y$ are independent random variables, each having the probability distribution $\bind(n,1/2)$,
  then
  $$
\frac1{\sqrt{\pi n}}\of{1-\frac1{8n}+O\of{\frac1{n^2}}} \le \prob{X=Y} < \frac1{\sqrt{\pi n}}.
  $$
\end{lemma}

\begin{proof}
Note that
  $$
\prob{X=Y}=\sum_{k=0}^n\of{{n\choose k}2^{-n}}^2=2^{-2n}\sum_{k=0}^n{n\choose k}^2=2^{-2n}{2n\choose n},
$$
where the last equality is a special case of Vandermonde's convolution.
The lemma now follows from the well-known bounds for the central binomial coefficient, namely the upper bound
  \begin{equation}
    \label{eq:central-coef}
{2n\choose n}<\frac{2^{2n}}{\sqrt{\pi n}}
\end{equation}
and the asymptotics ${2n\choose n}=\frac{2^{2n}}{\sqrt{\pi n}}\of{1-\frac1{8n}+O(\frac1{n^2})}$; see, e.g.,~\cite{Luke69}.
\end{proof}

A version of Lemma \ref{lem:binomial} could alternatively be proved by considering
the random variable $X+(n-Y)\sim\bind(2n,1/2)$ and estimating the probability that $X+(n-Y)=n$
with the help of the following classical result (which we state in the form restricted to our purposes).

\begin{lemma}[The local de Moivre--Laplace theorem; see, e.g.,~{\cite[Theorem 1, Chapter VII.3]{Feller}}]\label{lem:ML}
  Let $\phi(t)=\frac1{\sqrt{2\pi}}e^{-t^2/2}$ denote the density function of the standard normal distribution.
  Fix a sequence of positive reals $K_n=o(n^{2/3})$.
  Let $X\sim\mathrm{Bin}(n,1/2)$ and suppose that $n$ is even. Then
%   $$
% \prob{X=\frac n2+k}=(1+o(1))\frac{e^{-2k^2/n}}{\sqrt{\pi n/2}}
% $$
  $$
\prob{X=\frac n2+k}=(1+o(1))\frac{2\phi(2k/\sqrt{n})}{\sqrt{n}}=(1+o(1))\frac{e^{-2k^2/n}}{\sqrt{\pi n/2}}
$$
where the infinitesimal $o(1)$ approaches 0, as $n\to\infty$, uniformly over $k\in[-K_n,K_n]$.
\end{lemma}

We will also use a celebrated anti-concentration inequality for linear combinations of independent Bernoulli random variables.

\begin{lemma}[Erd\H{o}s~\cite{E}, Littlewood and Offord~\cite{LO}]
  Let $X=a_1X_1+\ldots+a_nX_n$, where $a_1,\ldots,a_n$ are non-zero reals and $X_1,\ldots,X_n$ are
  independent $\bind(1,1/2)$ random variables. Then $$\sup_{z\in\mathbb{R}}\prob{X=z}<n^{-1/2}.$$
\label{th:ELO}
\end{lemma}

\subsection{Proof of Theorem \ref{thm:3paths}: Part 1}\label{proof:1_1}

For a vertex $i\in[n]$, recall that $\mathbf{w}^G_3(i)=(w^G_1(i),w^G_2(i),w^G_3(i))$.
By the union bound,
$$
\prob{\mathbf{w}^G_3(i)=\mathbf{w}^G_3(j)\text{ for some }i,j}\le
\sum_{i,j}\prob{\mathbf{w}^G_3(i)=\mathbf{w}^G_3(j)}=
{n\choose 2}\prob{\mathbf{w}^G_3(1)=\mathbf{w}^G_3(2)}.
$$
Therefore, it suffices to prove that
\begin{equation}
  \label{eq:w3short}
\prob{\mathbf{w}^G_3(1)=\mathbf{w}^G_3(2)}=O(n^{-5/2}).  
\end{equation}

Let $N_H(v)$ denote the neighborhood of a vertex $v$ in a graph $H$.
Given two sets $U_1\subset[n]\setminus\{1\}$ and $U_2\subset[n]\setminus\{2\}$,
let $G'=G'(U_1,U_2)$ be the random graph $G$ subject to the conditions $N_G(1)=U_1$ and $N_G(2)=U_2$. 
In other terms, $G'$ is a random graph on $[n]$ chosen equiprobably among all graphs
satisfying these conditions.
Let $w'_k(i)=w^{G'}_k(i)$ denote the number of walks of length $k$ emanating from $i$ in $G'$
(the dependence of $w'_k(i)$ on the pair $U_1,U_2$ will be dropped for the sake of notational convenience). Define
$$
p(U_1,U_2)=\probb{\sum_{i\in U_1}w'_1(i)=\sum_{i\in U_2}w'_1(i)\text{ and }\sum_{i\in U_1}w'_2(i)=\sum_{i\in U_2}w'_2(i)}.
$$
We have
\begin{multline}
  \prob{\mathbf{w}^G_3(1)=\mathbf{w}^G_3(2)}\\ =
  \sum_{U_1,U_2\,:\,|U_1|=|U_2|}
  \cprob{\mathbf{w}^G_3(1)=\mathbf{w}^G_3(2)}{N_G(1)=U_1,\ N_G(2)=U_2}\\
  \mbox{}\hspace{30mm}\times\prob{N_G(1)=U_1,\, N_G(2)=U_2}\\
  =\sum_{U_1,U_2\,:\,|U_1|=|U_2|}p(U_1,U_2)\times\prob{N_G(1)=U_1,\, N_G(2)=U_2}.\label{eq:w3long}
\end{multline}
Note first that
\begin{multline*}
  \sum_{|U_1|=|U_2|}\probb{N_G(1)=U_1,\,N_G(2)=U_2}=\probb{|N_G(1)|=|N_G(2)|}\\
  =\probb{|N_G(1)\setminus\{2\}|=|N_G(2)\setminus\{1\}|}=O(n^{-1/2})
\end{multline*}
by Lemma \ref{lem:binomial} because $|N_G(1)\setminus\{2\}|\sim\bind(n-2,1/2)$
and $|N_G(2)\setminus\{1\}|\sim\bind(n-2,1/2)$ are independent binomial random variables.
This allows us to derive \refeq{w3short} from \refeq{w3long} if we prove that
\begin{equation}
  \label{eq:pUU}
p(U_1,U_2)=O(n^{-2})  
\end{equation}
for the neighborhood sets $U_1$ and~$U_2$.

In fact, we do not need to prove \refeq{pUU} for all pairs $U_1,U_2$ because the
contribution of some of them in \refeq{w3long} is negligible. Indeed, set $\varepsilon(n)=n^{-1/4}$.
Note that $|N_G(j)|\sim\bind(n-1,1/2)$ for $j=1,2$ and $|(N_G(1)\cap N_G(2))\setminus\{1,2\}|\sim\bind(n-2,1/4)$.
By the Chernoff bound (see Lemma \ref{lem:chernoff}), we have
$(1/2-\varepsilon(n))n\le |N_G(j)|\le (1/2+\varepsilon(n))n$ for $j=1,2$ and
$(1/4-\varepsilon(n))n\le |N_G(1)\cap N_G(2)|\le (1/4+\varepsilon(n))n$ with probability
$1-e^{-\Omega(\sqrt n)}$. Call a pair $U_1,U_2$ \emph{standard} if $|U_j|$ for $j=1,2$ and $|U_1\cap U_2|$
are in the above ranges. Thus, all non-standard pairs make a negligible contribution in \refeq{w3long},
and we only have to prove \refeq{pUU} for each standard pair~$U_1,U_2$.

For a graph $H$ and a subset $U\subset V(H)$, let $E_H(U)$ denote the set of edges of
$H$ with at least one vertex in $U$. Given two sets of edges $\mathcal{E}_1$ and $\mathcal{E}_2$
incident to the vertices in $U_1\setminus\{2\}$ and $U_2\setminus\{1\}$ respectively, let
$G''=G''(U_1,U_2,\mathcal{E}_1,\mathcal{E}_2)$ be the random graph $G'$ subject to the conditions
\begin{equation}
  \label{eq:EE}
E_{G'}(U_1\setminus\{2\})=\mathcal{E}_1\text{ and }E_{G'}(U_2\setminus\{1\})=\mathcal{E}_2.  
\end{equation}
Let $w''_k(i)=w^{G''}_k(i)$ denote the number of walks of length $k$ emanating from $i$ in $G''$
(the dependence of $w''_k(i)$ on $U_1,U_2,\mathcal{E}_1,\mathcal{E}_2$ is dropped for notational simplicity).

Note that Condition \eqref{eq:EE} determines whether or not $\sum_{i\in U_1}w_1'(i)=\sum_{i\in U_2}w_1'(i)$.
If this equality is true, we call the pair $\mathcal{E}_1,\mathcal{E}_2$ \emph{balanced}.
Using this terminology and notation, we can write
\begin{multline}
  p(U_1,U_2)=
  \sum_{\mathcal{E}_1,\mathcal{E}_2:\text{ balanced}}
  \probb{E_{G'}(U_1\setminus\{2\})=\mathcal{E}_1,\,E_{G'}(U_2\setminus\{1\})=\mathcal{E}_2}\\
  \times
  \probb{\sum_{i\in U_1}w''_2(i)=\sum_{i\in U_2}w''_2(i)}.\label{eq:pUUEE}
\end{multline}

We first show that
\begin{equation}
  \label{eq:UUEE}
  \sum_{\mathcal{E}_1,\mathcal{E}_2:\text{ balanced}}
  \probb{E_{G'}(U_1\setminus\{2\})=\mathcal{E}_1,\,E_{G'}(U_2\setminus\{1\})=\mathcal{E}_2}=O\of{n^{-1}}.
\end{equation}
Note that the sum in the left hand side of \refeq{UUEE} is equal to the probability that
$\sum_{i\in U_1}w_1'(i)=\sum_{i\in U_2}w_1'(i)$. This equality is equivalent to
$\sum_{i\in U_1\setminus (U_2\cup\{2\})}w_1'(i)=\sum_{i\in U_2\setminus (U_1\cup\{1\})}w_1'(i)$,
which in its turn is true if and only if $U_1\setminus (U_2\cup\{2\})$ and $U_2\setminus (U_1\cup\{1\})$
send the same number of edges to $[n]\setminus[(U_1\cup U_2\cup\{1,2\})\setminus(U_1\cap U_2)]$.
Since the pair $U_1,U_2$ is standard, these numbers are independent binomial random variables
with $\Theta(n^2)$ trials. Equality \refeq{UUEE} now follows by Lemma~\ref{lem:binomial}.

We now can derive Equality \refeq{pUU} from Equalities \refeq{pUUEE} and \refeq{UUEE} by proving that
\begin{equation}
  \label{eq:ww}
\probb{\sum_{i\in U_1}w''_2(i)=\sum_{i\in U_2}w''_2(i)}=O(n^{-1})
\end{equation}
for each potential pair $\mathcal{E}_1,\mathcal{E}_2$. Again, it is enough to
do this only for most probable pairs whose contribution in \refeq{pUUEE} is overwhelming.
Specifically, %consider the set $\mathcal{E}^*$ of all pairs $(\mathcal{E}_1,\mathcal{E}_2)$ such that 
 let $w'_2(u,v)$ denote the number of all paths of length $2$ between $u$ and $v$ in $G'$
and define
$$
\Delta(i,j)=(w'_2(i,1)+w'_2(j,1))-(w'_2(i,2)+w'_2(j,2)).
$$
Note that the numbers $w'_2(i,1),w'_2(j,1),w'_2(i,2),w'_2(j,2)$ and, hence, the numbers $\Delta(i,j)$
are completely determined by specifying $E_{G'}(U_1\setminus\{2\})=\mathcal{E}_1$ and
$E_{G'}(U_2\setminus\{1\})=\mathcal{E}_2$. We call a pair $\mathcal{E}_1,\mathcal{E}_2$
\emph{standard} if the number of pairs $\{i,j\}$ such that $\Delta(i,j)=0$ is at most $n^{7/4}$ as
$i\neq j$ ranges over the pairs of vertices from $[n]\setminus (U_1\cup U_2\cup \{1,2\})$.
The following fact shows that it will be enough to prove \refeq{ww}
for each standard pair~$\mathcal{E}_1,\mathcal{E}_2$.

\begin{claim}
If a pair $U_1,U_2$ is standard, then
$$
|\Set{\{i,j\}}{i,j\in [n]\setminus (U_1\cup U_2\cup \{1,2\}),\,i\ne j,\,\Delta(i,j)=0}|=O(n^{3/2})
$$
with probability $1-O(n^{-2})$.  
\end{claim}

\begin{subproof}
Let $u_1=|U_1|$, $u_2=|U_2|$, and $u=|U_1\cap U_2|$. Note that, for different $i,j\in [n]\setminus (U_1\cup U_2\cup \{1,2\})$,
$$
\Delta(i,j)=|N_{G'}(i)\cap (U_1\setminus U_2)|+|N_{G'}(j)\cap (U_1\setminus U_2)|-|N_{G'}(i)\cap (U_2\setminus U_1)|-|N_{G'}(j)\cap (U_2\setminus U_1)|.
$$
Fix $i\notin U_1\cup U_2\cup \{1,2\}$.
Expose all edges from $i$ to $(U_1\setminus U_2)\cup (U_2\setminus U_1)$ and set
%$n_i^+:=|N_{G'}(i)\cap (U_1\setminus U_2)|$ and $n_i^-:=|N_{G'}(i)\cap (U_2\setminus U_1)|$. 
$c_i=|N_{G'}(i)\cap (U_2\setminus U_1)|-|N_{G'}(i)\cap (U_1\setminus U_2)|$.
For $j\notin U_1\cup U_2\cup \{1,2,i\}$, consider the random variable
$\xi_j=|N_{G'}(j)\cap (U_1\setminus U_2)|-|N_{G'}(j)\cap (U_2\setminus U_1)|$.
We have $\Delta(i,j)=0$ if and only if $\xi_j=c_i$.
The random variables $|N_{G'}(j)\cap (U_1\setminus U_2)|\sim \bind(u_1-u,1/2)$ and $|N_{G'}(j)\cap (U_2\setminus U_1)|\sim\bind(u_2-u,1/2)$ are independent. Since the pair $U_1,U_2$ is standard, Lemma~\ref{th:ELO} implies that $\prob{\xi_j=c_i}=O(n^{-1/2})$. Since the random variables $\xi_j$ are independent and equally distributed, the Chernoff bound (see Lemma \ref{lem:chernoff}) implies that the number of $j$ with $\xi_j=c_i$ is $O(n^{1/2})$ with probability $1-O(n^{-3})$.
%The four terms in the right hand side are independent random variables $\bind(u_1-u,1/2)$, $\bind(u_1-u,1/2)$, $\bind(u_2-u,1/2)$, $\bind(u_2-u,1/2)$ respectively. Since $N-\bind(N,p)\sim\bind(N,1-p)$, we conclude that
%$\Delta(i,j) \sim 2u-2u_2+\bind(2u_1+2u_2-4u,1/2)$. The Chernoff bound (see Lemma \ref{lem:chernoff}) implies that, for each pair $i,j$,
%the inequalities
%\begin{multline*}
%2u-2u_2+ (u_1+u_2-2u)\of{1-\frac{\sqrt{2\ln n}}{\sqrt{u_1+u_2-2u}}}\\ \le
%\Delta(i,j)\le 2u-2u_2+ (u_1+u_2-2u)\of{1+\frac{\sqrt{2\ln n}}{\sqrt{u_1+u_2-2u}}} 
%\end{multline*}
%are violated with probability at most $O(n^{-8})$.
This proves that, for an arbitrarily fixed $i$, the number of $j$ with $\Delta(i,j)=0$ is $O(n^{1/2})$ with probability $1-O(n^{-3})$. By the union bound, the number of pairs $\{i,j\}$ with $\Delta(i,j)=0$ is $O(n^{3/2})$ with probability $1-O(n^{-2})$.
\end{subproof}

It remains to prove \refeq{ww} for a fixed standard pair~$\mathcal{E}_1,\mathcal{E}_2$.
The equality $\sum_{i\in U_1}w''_2(i)=\sum_{i\in U_2}w''_2(i)$ is true
exactly when the number of walks of length 3 in $G''$ starting from $k$
is the same for $k=1$ and for $k=2$. Let $\gamma_k$ denote the number of such walks
with at least 2 vertices inside $U_1\cup U_2\cup\{1,2\}$. Note that the numbers $\gamma_1$ and $\gamma_k$
are determined by $U_1,U_2,\mathcal{E}_1,\mathcal{E}_2$. The remaining walks have at
least two vertices outside $U_1\cup U_2\cup\{1,2\}$.
For $i,j\notin U_1\cup U_2\cup\{1,2\}$,
let $e''_{i,j}$ be the indicator random variable of the presence of the edge $\{i,j\}$ in $G''$.
Then the number of walks of the last kind starting from $k=1,2$
is equal to $\sum_{i,j\notin U_1\cup U_2\cup \{1,2\}}  (w'_2(i,k)+w'_2(j,k))\,e''_{ij}$.
The equality $\sum_{i\in U_1}w''_2(i)=\sum_{i\in U_2}w''_2(i)$ can, therefore, be rewritten as
\begin{equation}
  \label{eq:gsgs}
\sum_{i,j\notin U_1\cup U_2\cup \{1,2\}}  (w'_2(i,1)+w'_2(j,1)-w'_2(i,2)-w'_2(j,2))\,e''_{ij}=
\gamma_2-\gamma_1.  
\end{equation}
Since $\mathcal{E}_1,\mathcal{E}_2$ is a standard pair, the coefficient $\Delta(i,j)$ in front of $e''_{ij}$ in~\eqref{eq:gsgs} equals 0 for at most $n^{7/4}$ of all $\Theta(n^2)$ pairs $\{i,j\}$. Therefore, Equality~\eqref{eq:gsgs} is satisfied with probability $O(n^{-1})$ due to Lemma~\ref{th:ELO}. This completes the proof of Equality \refeq{ww} and of the whole theorem.

\begin{remark}
  The probability bound in Theorem~\ref{thm:3paths} cannot be significantly improved because
  $\prob{\mathbf{w}^G_3(1)=\mathbf{w}^G_3(2)}=n^{-\Theta(1)}$. To see this, note first that
  $\prob{w^G_1(1)=w^G_1(2)}=\Theta(n^{-1/2})$ as a consequence of Lemma~\ref{lem:binomial}.
  For a standard pair $U_1,U_2$ with $|U_1|=|U_2|$, we can similarly show that
  $\prob{\sum_{i\in U_1}w_1'(i)=\sum_{i\in U_2}w_1'(i)}=\Theta(n^{-1})$. This implies that $\prob{\of{w^G_1(1),w^G_2(1)}=\of{w^G_1(2),w^G_2(2)}}=\Theta(n^{-3/2})$.  Showing a polynomial lower bound for $\prob{\of{w^G_1(1),w^G_2(1),w^G_3(1)}\allowbreak=\of{w^G_1(2),w^G_2(2),w^G_3(2)}}$ is a more delicate issue. One of the possible approaches is to observe that, subject to standard $U_1,U_2$ and $\mathcal{E}_1,\mathcal{E}_2$, with sufficiently high probability there exists an interval $I=\{a,a+1,\ldots,b\}$ of consecutive, either positive or negative integers such that $|a|=\Theta(\sqrt{n})$, $|I|=\Theta(\sqrt{n})$, and, for every $x\in I$, at least $\Omega(n^{3/2})$ coefficients $\Delta(i,j)=w'_2(i,1)+w'_2(j,1)-w'_2(i,2)-w'_2(j,2)$ are equal to $x$.
Let $Q$ denote the set of pairs $i,j\notin U_1\cup U_2\cup \{1,2\}$ with $\Delta(i,j)\notin I$.
Equality \refeq{gsgs} can now be rewritten as
\begin{equation}
\label{eq:gsgsgs}
 \sum_{x\in I} x\,\eta_x=\gamma_1-\gamma_2-\sum_{\{i,j\}\in Q}  \Delta(i,j)\,e''_{ij},
\end{equation} 
where $\eta_x\sim\mathrm{Bin}(\Theta(n^{3/2}),1/2)$. Let us expose all edges between the pairs $\{i,j\}$ in $Q$. The typical absolute value of the right hand side of \refeq{gsgsgs} is $O(n^{3/2})$. With some extra work one can show that, for typical pairs $U_1,U_2$ and $\mathcal{E}_1,\mathcal{E}_2$, Equality~\eqref{eq:gsgsgs} is satisfied with probability $n^{-\Theta(1)}$, which is enough for obtaining a desired lower bound.
 %   due to Lemma~\ref{lem:ML}, it remains to ensure that $\gamma_2-\gamma_1$ is divisible by $x$. The latter happens with probability $n^{-\Omega(1)}$: subject to the 2-balls around vertices 1 and 2, we have an interval of length at least $100n$ for the values of $\gamma_1-\gamma_2$ that are reachable with probability $\Omega(n^{-1})$. As easily seen, this is enough for obtaining a desired lower bound.
\end{remark}

\subsection{Proof of Theorem \ref{thm:3paths}: Part 2}\label{proof:1_2}

%For a vertex $i\in[n]$, recall that $\mathbf{w}^G_2(i)=(w^G_1(i),w^G_2(i))$.
We have to show that with high probability there are vertices $i\neq j$ such that $\mathbf{w}^G_2(i)=\mathbf{w}^G_2(j)$.
Let $X$ be the number of all such vertex pairs. Thus, we have to prove that $X>0$ with high probability
or, equivalently, that $\prob{X=0}=o(1)$. By Chebyshev's inequality,
$$
\prob{X=0}\le\frac{\var(X)}{(\mean X)^2}.
$$
Therefore, it suffices to prove that $\var(X)=o\of{(\mean X)^2}$, where the variance $\var(X)$ and
the mean $\mean X$ are seen as functions of $n$. We split the proof into two parts.
First, we determine the exact asymptotics of $\mean X$ and, second,
we establish an upper bound for $\var(X)$ showing that it is sufficiently small
if compared to $(\mean X)^2$.

\paragraph*{Estimation of $\mean X$.}
We claim that
\begin{equation}
 \mean X=\frac{1+o(1)}{\pi}\,\sqrt{n}.
\label{eq:exp_equiv_pairs}
\end{equation}
Let $X_{i,j}$ denote the indicator random variable of the event that $\mathbf{w}^G_2(i)=\mathbf{w}^G_2(j)$.
Thus, $X=\sum_{1\le i<j\le n}X_{i,j}$. Since all $X_{i,j}$ are identically distributed,
we can fix a pair of distinct vertices $i$ and $j$ arbitrarily and conclude, by the linearity of the expectation, that
$$
 \mean X={n\choose 2}\mean X_{i,j}={n\choose 2}\prob{X_{i,j}=1}.
 $$
Equality \refeq{exp_equiv_pairs} is, therefore, equivalent to the equality
\begin{equation}
  \label{eq:Xij}
  \prob{X_{i,j}=1}=\frac{2+o(1)}{\pi n^{3/2}},  
\end{equation}
which we will now prove.

Let $\xi_i$ and $\xi_j$ be the number of neighbors of the vertices $i$ and $j$ in $[n]\setminus\{i,j\}$ respectively.
Note that $\xi_i$ and $\xi_j$ are independent random variables with distribution $\bind(n-2,1/2)$.
By Lemma \ref{lem:binomial}, we have
\begin{equation}
  \label{eq:w1w1} 
\prob{w^G_1(i)=w^G_1(j)}=\prob{\xi_i=\xi_j}=\frac{1+o(1)}{\sqrt{\pi n}}.
\end{equation}
It follows that
\begin{equation}
  \label{eq:Xij-cond}
\prob{X_{i,j}=1}=\frac{1+o(1)}{\sqrt{\pi n}}\,\cprob{w^G_2(i)=w^G_2(j)}{w^G_1(i)=w^G_1(j)},  
\end{equation}
and our task reduces to estimating the conditional probability in the right hand side.

Let $N_G[v]=N_G(v)\cup\{v\}$ denote the closed neighborhood of a vertex $v$ in $G$.
% of equal sizes $(1/4+o(1))n$
For two disjoint sets $U,V\subseteq[n]\setminus\{i,j\}$, define the event 
$\mathcal{B}_{U,V}=\mathcal{B}_{U,V}(i,j)$ by the conditions
\begin{equation}
  \label{eq:BUV}
N_G(i)\setminus N_G[j]=U\text{ and }N_G(j)\setminus N_G[i]=V.  
\end{equation}
%Note that with exponentially small probability $|N(i)\setminus N(j)|\neq (1/4+o(1))n$ or $|N(j)\setminus N(i)|\neq (1/4+o(1))n$.
Note that
\begin{multline}
\cprob{w^G_2(i)=w^G_2(j)}{w^G_1(i)=w^G_1(j)}\\=
\sum_{|U|=|V|}\cprob{w^G_2(i)=w^G_2(j)}{\mathcal{B}_{U,V}}
  \cprob{\mathcal{B}_{U,V}}{w^G_1(i)=w^G_1(j)},  \label{eq:cprob}
\end{multline}
where the summation goes over all $U,V$ of equal size.
Set $\varepsilon(n)=n^{-1/3}$.
Note that $|N_G(i)\setminus N_G[j]|\sim\bind(n-2,1/4)$ and, by symmetry, $|N_G(j)\setminus N_G[i]|$
has the same distribution.
By the Chernoff bound (Lemma \ref{lem:chernoff}), both numbers belong to the interval
$[(1/4-\varepsilon(n))n,(1/4+\varepsilon(n))n]$ with probability
$1-e^{-\Omega(\sqrt[3] n)}$. Call a pair $U,V$ \emph{standard} if $|U|=|V|$
is in the same range. Thus, all non-standard pairs make a negligible contribution in the sum in \refeq{cprob}.
More precisely, the restriction of this sum to the non-standard pairs
is bounded by
\begin{multline}
  \sum_{U,V\text{ non-standard}} \cprob{\mathcal{B}_{U,V}}{w^G_1(i)=w^G_1(j)}\\
  \le
  \frac{\prob{N_G(i)\setminus N_G[j],\,N_G(j)\setminus N_G[i]\text{ is a non-standard pair}}}{\prob{w^G_1(i)=w^G_1(j)}}\\
  =O(\sqrt n)e^{-\Omega(\sqrt[3]n)}=e^{-\Omega(\sqrt[3]n)},\label{eq:nonstand}
\end{multline}
where for the last but one equality we use Equality \refeq{w1w1}. 
To complete the estimation of $\mean X$, it suffices to prove the following fact.

\begin{claim}\label{cl:BUV}
For each standard pair~$U,V$,
$$
 \cprob{w^G_2(i)=w^G_2(j)}{\mathcal{B}_{U,V}}=\frac{2+o(1)}{n\sqrt{\pi}},
$$
where the infinitesimal $o(1)$ is the same for all standard pairs. 
\end{claim}

\noindent
Indeed, from Equality \refeq{cprob} we obtain
\begin{multline*}
\cprob{w^G_2(i)=w^G_2(j)}{w^G_1(i)=w^G_1(j)}\\
 =\sum_{U,V\text{ standard}}\cprob{w^G_2(i)=w^G_2(j)}{\mathcal{B}_{U,V}}
\cprob{\mathcal{B}_{U,V}}{w^G_1(i)=w^G_1(j)}+e^{-\Omega(\sqrt[3] n)}\\
=\frac{2+o(1)}{n\sqrt{\pi}}\sum_{U,V\text{ standard}}\cprob{\mathcal{B}_{U,V}}{w^G_1(i)=w^G_1(j)}+e^{-\Omega(\sqrt[3] n)}\\
=\frac{2+o(1)}{n\sqrt{\pi}}(1-e^{-\Omega(\sqrt[3] n)})+e^{-\Omega(\sqrt[3] n)}=
\frac{2+o(1)}{n\sqrt{\pi}}.
\end{multline*}
We here used Estimate \refeq{nonstand} for the first equality, Claim \ref{cl:BUV} for the second equality,
and once again Estimate \refeq{nonstand} for the third equality.
Plugging it into Equality \refeq{Xij-cond}, we get \refeq{Xij} and hence also \refeq{exp_equiv_pairs}.
It remains to prove Claim~\ref{cl:BUV}.

\begin{subproof}
% [Proof of Claim~\ref{cl:BUV}]
Let 
\begin{itemize}
\item $\xi_{U,V}$ be the number of edges between $U$ and $[n]\setminus (U\cup V\cup\{i,j\})$, and
%\item $\xi_V$ be the number of edges between $V$ and $[n]\setminus (U\cup V\cup\{i,j\})$,
\item $\xi_U$ be the number of edges induced by $U$.
%\item $\xi'_V$ be the number of edges induced by $V$.
\end{itemize}
Define $\xi_{V,U}$ and $\xi_V$ in the same way. Set 
$$
N=|U|(n-2|U|-2)\text{ and }M={|U|\choose 2}.
$$
Since $U,V$ is a standard pair, $N=(1/8+o(1))n^2$ and $M=(1/32+o(1))n^2$.
Here and below, the infinitesimals do not depend on a particular standard pair $U,V$. 
Note that $\xi_{U,V},\,\xi_{V,U}\sim\bind(N,1/2)$ and $\xi_U,\xi_V\sim\bind(M,1/2)$,
and that the four random variables are independent. The relevance of these random variables lies in
the equality
\begin{equation}
  \label{eq:w2BUV}
 \cprob{w^G_2(i)=w^G_2(j)}{\mathcal{B}_{U,V}}=
 \prob{\xi_{U,V}+2\xi_U=\xi_{V,U}+2\xi_V}.  
\end{equation}
To see this equality, assume that $\mathcal{B}_{U,V}$ occurs.
Notice that the number of 2-walks from $i$ and $j$ via one of their common neighbors
is the same for these two vertices, and the same holds true for the number of 2-walks containing
an edge between $U$ and $V$. The number of remaining  2-walks from $i$ is clearly equal to $\xi_{U,V}+2\xi_U$,
and this number is equal to $\xi_{V,U}+2\xi_V$ for $j$.

Rewriting the equality $\xi_{U,V}+2\xi_U=\xi_{V,U}+2\xi_V$ as $\xi_{U,V}-\xi_{V,U}=2(\xi_V-\xi_U)$,
we can rephrase it as follows: There is an integer $x$ in the range from $-M$ to $M$ such that
$\xi_V+(M-\xi_U)=M+x$ and $\xi_{U,V}+(N-\xi_{V,U})=N+2x$. Thus,
\begin{equation}
  \label{eq:NM}
\cprob{w^G_2(i)=w^G_2(j)}{\mathcal{B}_{U,V}}=
\sum_{x=-M}^M
\prob{\xi_{U,V}+(N-\xi_{V,U})=N+2x} \prob{\xi_V+(M-\xi_U)=M+x}.
\end{equation}
The advantage of considering these two
random variables is that $\xi_V+(M-\xi_U)\sim\bind(2M,1/2)$ and $\xi_{U,V}+(N-\xi_{V,U})\sim\bind(2N,1/2)$
and that they are independent. By Chernoff's bound,
$$
\prob{\xi_V+(M-\xi_U)\notin(M-M^{5/9},M+M^{5/9})}=e^{-\Omega(M^{1/9})}=e^{-\Omega(n^{2/9})}.
$$
This shows that the restriction of the sum in Equality \refeq{NM}
to $x\notin(-M^{5/9},M^{5/9})$ is bounded by $e^{-\Omega(n^{2/9})}$. For $x\in(-M^{5/9},M^{5/9})$,
the local de Moivre--Laplace theorem (Lemma \ref{lem:ML}) yields
$$
\prob{\xi_{U,V}+(N-\xi_{V,U})=N+2x}=(1+o(1))\frac{e^{-4x^2/N}}{\sqrt{\pi N}}
$$
and
$$
\prob{\xi_V+(M-\xi_U)=M+x}=(1+o(1))\frac{e^{-x^2/M}}{\sqrt{\pi M}}
$$
with the infinitesimal approaching 0 uniformly over $x$ in the specified range.
It follows that
\begin{multline}
\cprob{w^G_2(i)=w^G_2(j)}{\mathcal{B}_{U,V}}=
\sum_{-M^{5/9}<x<M^{5/9}}
(1+o(1))\frac{e^{-4x^2/N-x^2/M}}{\pi\sqrt{NM}}
+e^{-\Omega(n^{2/9})}\\
=\sum_{-M^{5/9}<x<M^{5/9}}
\frac{16+o(1)}{\pi n^2}e^{-64x^2/n^2}+e^{-\Omega(n^{2/9})}\\
=\frac{2+o(1)}{\sqrt{\pi}n}\sum_{-M^{5/9}<x<M^{5/9}}
\frac{1}{\sqrt{2\pi}}\exp\of{-\frac{128 x^2}{2 n^2}}\frac{8\sqrt{2}}{n}
+e^{-\Omega(n^{2/9})}\\
=\frac{2+o(1)}{\sqrt{\pi}n}\of{\int_{-\infty}^{\infty}\frac{1}{\sqrt{2\pi}}e^{-x^2/2}dx-o(1)}+e^{-\Omega(n^{2/9})}
=\frac{2+o(1)}{\sqrt{\pi}n},\label{eq:d_2_limit}
\end{multline}
completing the proof of the claim.
\end{subproof}

\paragraph*{Estimation of $\var(X)$.}
In view of Equality \eqref{eq:exp_equiv_pairs}, the proof of Part 2 of the theorem will be completed
if we show that
\begin{equation}
  \label{eq:var}
\var(X)=o(n).  
\end{equation}
Since $X=\sum_{i,j}X_{i,j}$, we have
$$
\var(X)=\sum_{i,j;\,i',j'}\cov(X_{i,j},X_{i',j'}),
$$
where the summation is over all possible pairs of 2-element sets $\{i,j\}$ and $\{i',j'\}$.
Note that there are ${n\choose 2}<n^2$ pairs with $\{i,j\}=\{i',j'\}$
and $n(n-1)(n-2)<n^3$ pairs with $|\{i,j\}\cap\{i',j'\}|=1$.
The number of the remaining pairs with $\{i,j\}\cap\{i',j'\}=\emptyset$ is bounded by $n^4$.
Therefore, Equality \eqref{eq:var} follows from the three estimates below.

\begin{claim}\label{cl:2}
$\cov(X_{i,j},X_{i,j})=O(n^{-3/2})$.  
\end{claim}

\begin{claim}\label{cl:1}
  If $i'=i$ and $j'\neq j$, then $\cov(X_{i,j},X_{i',j'})=O(n^{-3})$.
\end{claim}

\begin{claim}\label{cl:0}
  If $\{i,j\}\cap\{i',j'\}=\emptyset$, then $\cov(X_{i,j},X_{i',j'})=o(n^{-3})$.
\end{claim}

\bigskip

The rest of the proof consists of the proofs of these three claims.

\bigskip

\begin{subproofof}{2}
  We have
  $$
  \cov(X_{i,j},X_{i,j})=\mean X_{i,j}-(\mean X_{i,j})^2=\mean X_{i,j}(1-\mean X_{i,j})=
  O(n^{-3/2}),
  $$
  where the last equality is due to~\eqref{eq:Xij}.
\end{subproofof}

\begin{subproofof}{1}  
  It is enough to prove that
  \begin{equation}
    \label{eq:Xijj}
  \mean(X_{i,j}X_{i,j'})=O(n^{-3})    
  \end{equation}
  because, along with Equality \eqref{eq:Xij}, this will imply that
  $\cov(X_{i,j},X_{i,j'})=\mean(X_{i,j}X_{i,j'})-\mean X_{i,j}\mean X_{i,j'}=O(n^{-3})$

Note that $\mean(X_{i,j}X_{i,j'})=\prob{\mathbf{w}^G_2(i)=\mathbf{w}^G_2(j)=\mathbf{w}^G_2(j')}$.
We first estimate the probability that $w^G_1(i)=w^G_1(j)=w^G_1(j')$.  
Fix pairwise distinct $i,j,j'$ and expose the edges between these three vertices as well as
all edges from $i$. This determines $w^G_1(i)$ and, hence, also two integers $x$ and $x'$ such that
$w^G_1(i)=w^G_1(j)=w^G_1(j')$ if and only if $\xi=x$ and $\xi'=x'$ for $\xi=|N_G(j)\setminus\{i,j'\}|$
and $\xi'=|N_G(j')\setminus\{i,j\}|$. Since $\xi$ and $\xi'$ are independent random variables,
both with distribution $\bind(n-3,1/2)$, we conclude that
\begin{multline}
 \prob{w^G_1(i)=w^G_1(j)=w^G_1(j')}\le\of{\max_z\prob{\mathrm{Bin}(n-3,1/2)=z}}^2\\
 =\of{\frac{1}{2^{n-3}}{n-3\choose \lfloor(n-3)/2\rfloor}}^2\le\frac{2+o(1)}{\pi n},\label{eq:w1}
\end{multline}
where the last inequality is provided by Bound \eqref{eq:central-coef} if $n$ is even
and easily follows from \eqref{eq:central-coef} if $n$ is odd.

Next, we estimate the probability of $w^G_2(i)=w^G_2(j)=w^G_2(j')$ under the condition
$w^G_1(i)=w^G_1(j)=w^G_1(j')$. Let us split this condition into a set of pairwise disjoint
events. For two disjoint sets $A,A'\subseteq[n]\setminus\{i,j,j'\}$, let $\mathcal{D}_{A,A'}$
be the event that
\begin{itemize}
\item
  $i$, $j$, and $j'$ have the same degree,
\item
  $N_G(j)\setminus (N_G[i]\cup N_G[j'])=A$ and, similarly,
  % $A$ is exactly the set of those neighbors of $j$ in $[n]\setminus\{i,j,j'\}$
  % which are adjacent neither to $i$ nor or $j'$, and
\item
  $N_G(j')\setminus (N_G[i]\cup N_G[j])=A'$.
  % similarly,  $A'$ is the set of those neighbors of $j'$ in $[n]\setminus\{i,j,j'\}$
  % which are adjacent neither to $i$ nor or~$j$.
\end{itemize}
Note that
\begin{multline}
\cprob{w^G_2(i)=w^G_2(j)=w^G_2(j')}{w^G_1(i)=w^G_1(j)=w^G_1(j')}\\
=\sum_{A,A'}
\cprob{w^G_2(i)=w^G_2(j)=w^G_2(j')}{\mathcal{D}_{A,A'}}
\cprob{\mathcal{D}_{A,A'}}{w^G_1(i)=w^G_1(j)=w^G_1(j')}.\label{eq:w21}
\end{multline}
Set $\varepsilon(n)=n^{-1/4}$. Call a set $A$ \emph{standard} if $(1/8-\varepsilon(n))n\le|A|\le(1/8+\varepsilon(n))n$
and apply this definition as well to $A'$. The restriction of the sum \eqref{eq:w21} to the pairs $A,A'$
with not both $A$ and $A'$ being standard is bounded by
\begin{multline*}
\sum_{A\text{ or }A'\text{ non-standard}}\cprob{\mathcal{D}_{A,A'}}{w^G_1(i)=w^G_1(j)=w^G_1(j')}\\
\le\frac{\prob{N_G(j)\setminus (N_G[i]\cup N_G[j'])\text{ or }N_G(j')\setminus (N_G[i]\cup N_G[j])\text{ is non-standard}}}{\prob{w^G_1(i)=w^G_1(j)=w^G_1(j')}}\\
=e^{-\Omega(\sqrt n)}O(n)=e^{-\Omega(\sqrt n)},
\end{multline*}
where the last but one equality follows from Chernoff's bound and Estimate \eqref{eq:w1}.
To estimate the remaining part of the sum \eqref{eq:w21}, we fix a pair $A,A'$
with both $A$ and $A'$ being standard and estimate the conditional probability
$\cprob{w^G_2(i)=w^G_2(j)=w^G_2(j')}{\mathcal{D}_{A,A'}}$.

Expose all edges of $G$ except those that are entirely inside one of the sets $A$ and $A'$
in such a way that the condition $\mathcal{D}_{A,A'}$ is ensured. Note that $w^G_2(i)$ is therewith determined.
Moreover, the unexposed edges within $A$ contribute only in $w^G_2(j)$, and
the unexposed edges within $A'$ contribute only in $w^G_2(j')$. Let $\zeta$ be the number
of edges of the former kind, and $\zeta'$ be the number of edges of the latter kind.
It follows that if the equality $w^G_2(i)=w^G_2(j)=w^G_2(j')$ is not excluded by the
exposed edges, then there are integers $x$ and $x'$ such that this equality holds
exactly when $\zeta=x$ and $\zeta'=x'$.
Since $\zeta\sim\bind({|A|\choose2},1/2)$ and $\zeta'\sim\bind({|A'|\choose2},1/2)$
are independent, and both $A$ and $A'$ are standard, we conclude that
\begin{multline*}
\cprob{w^G_2(i)=w^G_2(j)=w^G_2(j')}{\mathcal{D}_{A,A'}}\\
  \le
  \max_{(1/8-\varepsilon(n))n\le a\le(1/8+\varepsilon(n))n}
  \of{\max_z\prob{\bind\left({a\choose 2},1/2\right)=z}}^2
  \le\frac{256+o(1)}{\pi n^2}.
\end{multline*}
This gives us an upper bound for the restriction of the sum \eqref{eq:w21}
to the pairs $A,A'$ with both $A$ and $A'$ being standard, implying that
$$
\cprob{w^G_2(i)=w^G_2(j)=w^G_2(j')}{w^G_1(i)=w^G_1(j)=w^G_1(j')}\le\frac{256+o(1)}{\pi n^2}.
$$
Combining this with Estimate \eqref{eq:w1}, we obtain the required Equality~\eqref{eq:Xijj}.
\end{subproofof}

\begin{subproofof}{0}
  It is enough to prove that
  \begin{equation}
    \label{eq:Xijij}
  \mean(X_{i,j}X_{i',j'})=\frac{4+o(1)}{\pi^2n^3}.
%  =(1+o(1))\mean X_{i,j}\mean X_{i',j'}    
\end{equation}
Once this equality is established, we will have
$\cov(X_{i,j},X_{i',j'})=\mean(X_{i,j}X_{i',j'})-\mean X_{i,j}\mean X_{i',j'}=o(n^{-3})$
by Equality \eqref{eq:Xij}.

Thus, we have to estimate the value of
\begin{multline*}
\mean(X_{i,j}X_{i',j'})=\prob{\mathbf{w}^G_2(i)=\mathbf{w}^G_2(j),\ \mathbf{w}^G_2(i')=\mathbf{w}^G_2(j')}\\
=\cprob{w^G_2(i)=w^G_2(j),\ w^G_2(i')=w^G_2(j')}{w^G_1(i)=w^G_1(j),\ w^G_1(i')=w^G_1(j')}\\
\times\prob{w^G_1(i)=w^G_1(j),\ w^G_1(i')=w^G_1(j')}.
\end{multline*}
It is easy to show that
\begin{multline}
  \prob{w^G_1(i)=w^G_1(j),\ w^G_1(i')=w^G_1(j')}\\
  = (1+o(1))\prob{w^G_1(i)=w^G_1(j)}\prob{w^G_1(i')=w^G_1(j')}=
 \frac{1+o(1)}{\pi n};\label{eq:w1-split}
\end{multline}
c.f.~Equality~\eqref{eq:w1w1}.
Indeed, after the edges between $\{i,j\}$ and $\{i',j'\}$ are exposed, the events $w^G_1(i)=w^G_1(j)$
and $w^G_1(i')=w^G_1(j')$ become independent, and we have to estimate the probability of each of them
separately. Consider, for example, the event $w^G_1(i)=w^G_1(j)$. Let $\xi_i$ and $\xi_j$ be the number
of neighbors of the vertices $i$ and $j$ in $[n]\setminus\{i,j,i',j'\}$ respectively.
If the number of exposed edges between $\{i\}$ and $\{i',j'\}$ is the same as between
$\{j\}$ and $\{i',j'\}$, then $w^G_1(i)=w^G_1(j)$ exactly when $\xi_i=\xi_j$,
and we can use Lemma \ref{lem:binomial} as we did it earlier to derive Equality \eqref{eq:w1w1}.
If the number of exposed edges from $i$ and from $j$ is different, then we can use the de Moivre--Laplace
limit theorem (Lemma \ref{lem:ML}). Suppose, for example, that $i$ is adjacent neither to $i'$
nor to $j'$ and that $j$ is adjacent to both $i'$ and $j'$. Then $w^G_1(i)=w^G_1(j)$ exactly when $\xi_i=\xi_j+2$.
The latter equality can be rewritten as $\xi_i+(n-4-\xi_j)=n-2$ and its probability is nothing else
as the probability of the event that $\bind(2(n-4),1/2)=(n-4)+2$.

In the rest of the proof, we estimate the probability of $w^G_2(i)=w^G_2(j)$ and $w^G_2(i')=w^G_2(j')$
under the condition that $w^G_1(i)=w^G_1(j)$ and $w^G_1(i')=w^G_1(j')$.
Similarly to the estimation of the conditional probability $\cprob{w^G_2(i)=w^G_2(j)}{w^G_1(i)=w^G_1(j)}$
above (see Equality \eqref{eq:cprob} and the subsequent argument), it suffices
to estimate the probability of $w^G_2(i)=w^G_2(j)$ and $w^G_2(i')=w^G_2(j')$
conditioned on then event $\mathcal{B}_{U,V}(i,j)\cap\mathcal{B}_{U',V'}(i',j')$
for two disjoint sets $U,V\subset[n]\setminus\{i,j\}$ and two disjoint sets $U',V'\subset[n]\setminus\{i',j'\}$
such that $|U|=|V|$ and $|U'|=|V'|$. Recall that the event $\mathcal{B}_{U,V}(i,j)$
is defined by \eqref{eq:BUV}. Accordingly to Chernoff's bound, it is enough to consider a
\emph{standard} quadruple $U,V,U',V'$ in the sense that the cardinalities of the four sets belong to the interval
$[(1/4-\varepsilon(n))n,(1/4+\varepsilon(n))n]$ and, additionally, the cardinalities
$|U\cap U'|$, $|U\cap V'|$, $|V\cap U'|$, and $|V\cap V'|$ belong to the interval
$[(1/16-\varepsilon(n))n,(1/16+\varepsilon(n))n]$, where we set $\varepsilon(n)=n^{-1/4}$.

Define the random variables $\xi_U$ and $\xi_{U,V}$, along with their analogs for the other sets,
exactly as in the proof of Claim~\ref{cl:BUV}. Similarly to Equality \eqref{eq:w2BUV}, we have
\begin{multline*}
  \cprob{w^G_2(i)=w^G_2(j),\ w^G_2(i')=w^G_2(j')}{\mathcal{B}_{U,V}(i,j)\cap\mathcal{B}_{U',V'}(i',j')}\\
  =\prob{\xi_{U,V}+2\xi_U=\xi_{V,U}+2\xi_V,\ \xi_{U',V'}+2\xi_{U'}=\xi_{V',U'}+2\xi_{V'}}
  =\prob{\Psi_1=0,\ \Psi_2=0},
\end{multline*}
where
$$
\Psi_1=\xi_{U,V}-\xi_{V,U}+2(\xi_U-\xi_V)\text{ and }\Psi_2=\xi_{U',V'}-\xi_{V',U'}+2(\xi_{U'}-\xi_{V'}).
$$
Taking into account Equality \eqref{eq:w1-split}, it is now clear that the desired Equality \eqref{eq:Xijij}
will be obtained if we prove that
\begin{equation}
\prob{\Psi_1=0,\ \Psi_2=0}=\frac{4+o(1)}{\pi n^2}
%\prob{\xi_{U,V}+2\xi_U=\xi_{V,U}+2\xi_V,\ \xi_{U',V'}+2\xi_{U'}=\xi_{V',U'}+2\xi_{V'}}\\
%=(1+o(1))\prob{\xi_{U,V}+2\xi_U=\xi_{V,U}+2\xi_V}\prob{\xi_{U',V'}+2\xi_{U'}=\xi_{V',U'}+2\xi_{V'}},
\label{eq:cond-ind}
\end{equation}
with infinitesimal $o(1)$ tending to 0 uniformly over all standard quadruples.

Note that $\prob{\Psi_k=0}=\frac{2+o(1)}{n\sqrt{\pi}}$ for each $k=1,2$
in view of Equalities \eqref{eq:w2BUV} and \eqref{eq:d_2_limit}.
Equality \eqref{eq:cond-ind}, therefore, implies that the events $\Psi_1=0$ and $\Psi_2=0$
are asymptotically independent. This may appear somewhat surprising as the random variables
$\xi_U$, $\xi_{U,V}$, $\xi_{U'}$, and $\xi_{U',V'}$ involved in $\Psi_1$ and $\Psi_2$ are quite
far from being independent (for example, $\xi_U$ and $\xi_{U'}$
are dependent because $|U\cap U'|=(1/16-o(1))n$). The forthcoming argument has to overcome
this complication.

Consider the random 2-dimensional vector $\Psi=(\Psi_1,\Psi_2)^\top$
and let $\Sigma=(\Sigma_{k,l})_{1\le k,l\le2}$ be the covariance matrix of $\Psi$.
Our nearest goal is to estimate the entries $\Sigma_{k,l}$ of this matrix.

We define four sets $S_1,S_2,S_3,S_4\subseteq{[n]\choose2}$ of vertex pairs as follows.
\begin{itemize}
\item
  $S_1$ is the set of all pairs $u,w$ with $u\in U$ and $w\notin U\cup V\cup\{i,j\}$.
\item
  $S_2$ is the set of all $v,w$ with $v\in V$ and $w\notin U\cup V\cup\{i,j\}$.
\item
$S_3={U\choose2}$.
\item
$S_4={V\choose2}$.
\end{itemize}
Let $S'_1,S'_2,S'_3,S'_4$ be defined similarly for $U'$ and $V'$.
Since $U,V,U',V'$ are standard, we have $|S_a|=(1/8+o(1))n^2$ for $a=1,2$
and $|S_a|=(1/32+o(1))n^2$ for $a=3,4$, and the similar equalities hold true
for the cardinalities $|S'_a|$.

For a vertex pair $e$, let $\iota_e$ denote the indicator random variable of
the event that $e$ is an edge in $G$. We define the random variable $\psi_e$ by
$$
\psi_e=\left\{
\begin{array}{rl}
  \iota_e, &\text{if }e\in S_1\\
  -\iota_e, &\text{if }e\in S_2\\
  2\iota_e, &\text{if }e\in S_3\\
  -2\iota_e, &\text{if }e\in S_4\\
  0, &\text{otherwise.}
\end{array}
\right.
$$
The random variable $\psi'_e$ is defined similarly for the sets $S'_1,S'_2,S'_3,S'_4$.
Note that
$$
\Psi_1=\sum_{e}\psi_e\text{ and }\Psi_2=\sum_{e}\psi'_e.
$$
As easily seen,
$$
\var\psi_e=\left\{
\begin{array}{rl}
  \frac14, &\text{if }e\in S_1\cup S_2\\
  1, &\text{if }e\in S_3\cup S_4,
\end{array}
\right.
$$
and the similar equalities hold true for $\psi'_e$.
Since the random variables $\psi_e$ are mutually independent, we have
\begin{equation}
  \label{eq:Sigma11}
\Sigma_{1,1}=\var\Psi_1=\sum_{e}\var\psi_e=\frac14|S_1|+\frac14|S_2|+|S_3|+|S_4|=\of{\frac18+o(1)}n^2
\end{equation}
and, similarly, $\Sigma_{2,2}=\var\Psi_2=(1/8+o(1))n^2$. Since $\psi_e$ and $\psi'_{e'}$ are independent
whenever $e\ne e'$, we have
$$
\Sigma_{1,2}=\cov(\Psi_1,\Psi_2)=\sum_{e}\cov(\psi_e,\psi'_e)=\sum_{a=1}^4\sum_{b=1}^4\sum_{e\in S_a\cap S'_b}\cov(\psi_e,\psi'_e).
$$
Note that, for each pair $1\le a,b\le4$, the covariance $\cov(\psi_e,\psi'_e)$ is the same for all $e\in S_a\cap S'_b$.
Denote this value by $c(a,b)$. Thus,
\begin{equation}
  \label{eq:sum-antipod}
\Sigma_{1,2}=\sum_{a=1}^4\sum_{b=1}^4 |S_a\cap S'_b|\,c(a,b).  
\end{equation}
Call two index pairs $a,b$ and $a^*,b^*$ \emph{antipodal} if $|S_a\cap S'_b|=(1+o(1))|S_{a^*}\cap S'_{b^*}|$
and $c(a,b)=-c(a^*,b^*)$. An important observation is that the sixteen index pairs can be split into eight antipodal pairs,
namely
\begin{align*}
(3,3)\text{ and } (3,4), \quad (1,1)\text{ and } (1,2), \quad (3,1)\text{ and } (3,2), \quad (4,1)\text{ and } (4,2),\\
(4,4)\text{ and } (4,3), \quad (2,2)\text{ and } (2,1), \quad (1,3)\text{ and } (2,3), \quad (1,4)\text{ and } (2,4).
\end{align*}
To show this, we consider three typical cases.

\smallskip

\textit{The pairs $(3,3)$ and $(3,4)$.} Note first that $S_3\cap S'_3={U\cap U'\choose2}$ and
$S_3\cap S'_4={U\cap V'\choose2}$. Since the quadruple $U,V,U',V'$ is standard, we have $|S_a\cap S'_b|=(2^{-9}+o(1))n^2$
for both $(a,b)=(3,3)$ and $(a,b)=(3,4)$. Furthermore, the definition of $\psi_e$ and $\psi'_e$ immediately
yields $c(3,3)=\frac12\cdot4-(\frac12\cdot2)^2=1$ and $c(3,3)=\frac12(-4)-(\frac12\cdot2)(-\frac12\cdot2)=-1$.
Note also that the pairs $(4,4)$ and $(4,3)$ form a quite similar configuration with virtually the same analysis.

\smallskip

\textit{The pairs $(1,1)$ and $(1,2)$.}
A pair $e=\{u,w\}$ in $S_1\cap S'_1$ can be (up to renaming $u$ and $w$) of one of the following two types:
either $u\in U\cap U'$ and $w\notin U\cup V\cup U'\cup V'$ or
$u\in U\setminus(U'\cup V')$ and $w\in U'\setminus(U\cup V)$.
Since the quadruple $U,V,U',V'$ is standard, we have
$|S_1\cap S'_1|=(\frac1{16}\cdot\frac14+(\frac18)^2+o(1))n^2=(\frac1{32}+o(1))n^2$.
The same estimate holds true for $|S_1\cap S'_2|$ by a similar argument.
Furthermore, $c(1,1)=\frac12-(\frac12)^2=\frac14$ while $c(1,2)=-\frac12-\frac12(-\frac12)=-\frac14$.
Similar considerations apply also to the pairs $(2,2)$ and $(2,1)$.

\smallskip

\textit{The pairs $(3,1)$ and $(3,2)$.}
Up to renaming $u$ and $w$, a pair $e=\{u,w\}$ belongs to $S_3\cap S'_1$ in the case that
$u\in U\cap U'$ and $w\in U\setminus(U'\cup V')$. It follows that
$|S_3\cap S'_1|=(\frac1{16}\cdot\frac18+o(1))n^2=(2^{-7}+o(1))n^2$.
The same estimate holds true as well for $|S_3\cap S'_2|$.
It is also easy to see that $c(3,1)=\frac12\cdot2-(\frac12\cdot2)\cdot\frac12=\frac12$
while $c(3,2)=-\frac12\cdot2-(\frac12\cdot2)\cdot(-\frac12)=-\frac12$.
The remaining antipodal pairs $(1,3)$ and $(2,3)$, $(4,1)$ and $(4,2)$, and $(1,4)$ and $(2,4)$
have the same form and are analyzed similarly.

\smallskip

In view of the fact that the index pairs are split into antipodes, Equality \eqref{eq:sum-antipod}
readily implies that $\Sigma_{1,2}=\Sigma_{2,1}=o(n^2)$. Thus,
$$
\Sigma=
\begin{pmatrix}
  \of{1/8+o(1)}n^2&o(n^2)\\
  o(n^2)&\of{1/8+o(1)}n^2
\end{pmatrix}.
$$
This shows, in particular, that the covariance matrix $\Sigma$ is invertible
and positive definite for all sufficiently large $n$.
Let $\Sigma^{1/2}$ denote the unique positive definite square root of $\Sigma$.
Using the known explicit formula for square roots of $2\times2$ matrices \cite{Levinger80},
for the inverse matrix $\Sigma^{-1/2}=(\Sigma^{1/2})^{-1}$ we obtain
$$
\Sigma^{-1/2}=\frac{1}{n}
\begin{pmatrix}
  2\sqrt2+o(1) & o(1) \\
  o(1) & 2\sqrt2+o(1)
\end{pmatrix}
=A+A',
$$
where
$$
A=\frac{1}{n}
\begin{pmatrix}
  2\sqrt2 & 0 \\
  0 & 2\sqrt2
\end{pmatrix}
\text{ and }
A'=\frac{1}{n}
\begin{pmatrix}
  o(1) & o(1) \\
  o(1) & o(1)
\end{pmatrix}.
$$
% can be viewed as a sum of $N=(1/8+o(1))n^2$ independent (not identically distributed) random vectors with values in $\{-3,-2,-1,0,1,2,3\}^2$. Indeed, let 
% $$
% N=\max\{|U|(n-2-|V|-|U|),|U'|(n-2-|V'|-|U'|)\}.
% $$
% Without loss of generality, let us assume that the maximum is achieved at $|U|(n-2-|V|-|U|)$. Let us now consider arbitrary orderings of all pairs of vertices 
% \begin{itemize}
% \item from $U\times [n]\setminus(U\cup V\cup\{i,j\})$, denote the $\ell$th pair by $e^{U,V}_{\ell}$ (we refer to these pairs of vertices as of pairs with the first type $s_1=1$); 
% \item from $V\times [n]\setminus(U\cup V\cup\{i,j\})$, denote the $\ell$th pair by $e^{V,U}_{\ell}$ (the first type $s_1=2$);
% \item inside $U$, denote the $\ell$th pair by $e^{U}_{\ell}$ (the first type $s_1=3$);
% \item inside $V$, denote the $\ell$th pair by $e^{V}_{\ell}$ (the first type $s_1=4$). 
% \end{itemize}
  
% For $\ell>{|U|\choose 2}$, we let $\mathbb{I}(e^{U}_{\ell}\in G_n)=0$, and same for $e^V_{\ell}$. Then, for any $\ell$, the first component of the $\ell$th vector equals $\mathbb{I}(e^{U,V}_{\ell}\in G_n)+2\mathbb{I}(e^{U}_{\ell}\in G_n)-\mathbb{I}(e^{V,U}_{\ell}\in G_n)-2\mathbb{I}(e^{V}_{\ell}\in G_n)$. The second components are defined in the same way but for $V'$ and $U'$, and the second type $s_2$ is defined accordingly. We note that the second components of $o(n^2)$ last vectors may equal $0$ almost surely when $|U'|(n-2-|V'|-|U'|)<|U|(n-2-|V|-|U|)$.

Consider the centered random variables $\bar\psi_e=\psi_e-\mean\psi_e$ and
$\bar\Psi_1=\Psi_1-\mean\Psi_1=\sum_{e}\bar\psi_e$.
Note that $\bar\Psi_1=\Psi_1$ because
$$
\mean\Psi_1=\sum_{e}\mean\psi_e=\of{\sum_{e\in S_1}\mean\psi_e+\sum_{e\in S_2}\mean\psi_e}
+\of{\sum_{e\in S_3}\mean\psi_e+\sum_{e\in S_4}\mean\psi_e}=0+0=0.
$$
Since the random variables $\bar\psi_e$ are uniformly bounded, we conclude by Equality \eqref{eq:Sigma11} that
$\sum_{e}\mean|\bar\psi_e|^3=o((\var\bar\Psi_1)^{3/2})$, showing that Lyapunov’s condition is fulfilled
(see \cite[Equality (27.16)]{Billingsley95}).
Using the Lyapunov central limit theorem (which applies to not necessarily identically distributed
random variables; see \cite[Theorem 27.3]{Billingsley95}) and noting that the same argument works
as well for $\Psi_2$, we establish the following convergences in distribution as $n$ goes to infinity:
\begin{equation}
 \frac{\Psi_1}{\sqrt{\var\Psi_1}}\stackrel{d}\to\eta\text{ and }
 \frac{\Psi_2}{\sqrt{\var\Psi_2}}\stackrel{d}\to\eta,
 \label{eq:CLT-single}
\end{equation}
where $\eta$ in both cases denotes a standard normal random variable.

From \eqref{eq:CLT-single}, we easily infer\footnote{%
  Fix $\epsilon>0$ and consider an arbitrarily small $\delta>0$.
  Choose $C$ so large that $\prob{\eta\notin(-C,C]}<\delta/2$.
  Take $N$ such that if $n\ge N$, then $|\alpha_n|<\epsilon/C$
  and the probability $\prob{\Psi_1/\sqrt{\var\Psi_1}\notin(-C,C]}$
  is at most $\delta/2$ apart from the probability $\prob{\eta\notin(-C,C]}$.
  Clearly, $\prob{|\alpha_n\Psi_1/\sqrt{\var\Psi_1}|>\epsilon}<\delta$ for all $n\ge N$.
}
two convergences in probability, which hold true
for every sequence of reals $\alpha_n$ approaching $0$ as $n$ increases:
$$
\frac{\alpha_n\Psi_1}{\sqrt{\var\Psi_1}}\stackrel{p}\to 0\text{ and }
\frac{\alpha_n\Psi_2}{\sqrt{\var\Psi_2}}\stackrel{p}\to 0.
$$
In view of \eqref{eq:Sigma11}, we conclude from here that
$$
\frac{\alpha_n\Psi_1}n\stackrel{p}\to 0\text{ and }
\frac{\alpha_n\Psi_2}n\stackrel{p}\to 0,
$$
for every sequence of reals $\alpha_n$ such that $\alpha_n\to0$ as $n\to\infty$.
Since all entries of the matrix $A'=\Sigma^{-1/2}-A$ are some infinitesimals $o(1/n)$,
this readily implies that
\begin{equation}\label{eq:AA}
A'\Psi\stackrel{p}\to 0, 
\end{equation}
where $0$ denotes the two-dimensional zero vector.

Turning back to the 2-dimensional vector $\Psi=(\Psi_1,\Psi_2)^\top$,
decompose it into the sum of independent centered random vectors
\begin{equation}
  \label{eq:Psi-decomp}
\Psi=\sum_{e}(\bar\psi_e,\bar\psi'_e)^\top.  
\end{equation}
The multivariate Lyapunov central limit theorem (e.g.,~\cite[Corollary 18.2]{BhattacharyaR10})
yields the convergence
\begin{equation}
   \Sigma^{-1/2}\Psi\stackrel{d}\to(\eta_1,\eta_2)^\top,
 \label{eq:CLT}
\end{equation}
where $\eta_1$ and $\eta_2$ are independent standard normal random variables.
Using Convergences \eqref{eq:CLT} and \eqref{eq:AA}, we conclude by the version of
Slutsky's theorem for random vectors (e.g., \cite[Lemma 2.8]{Vaart98}) that
$$
A \Psi\stackrel{d}\to(\eta_1,\eta_2)^\top
$$
that is,
\begin{equation}
\of{\frac{2\sqrt2\Psi_1}n,\frac{2\sqrt2\Psi_2}n} \stackrel{d}\to(\eta_1,\eta_2). 
% which immediately implies that
%  \frac{2\sqrt2\Psi_1}n\stackrel{d}\to\eta_1\text{ and }
%  \frac{2\sqrt2\Psi_2}n\stackrel{d}\to\eta_2.
\label{eq:integral_limit_th}
\end{equation}

Recall that our goal is estimating the probability $\prob{\Psi_1=0,\,\Psi_2=0}$.
This can be done by establishing a local version of \eqref{eq:integral_limit_th}.
There is a general way \cite{Mukhin} of converting an \emph{integral} limit theorem (ILT)
for independent random integer vectors
into a \emph{local limit} theorem (LLT), which perfectly suits our needs.
Below we state \cite[Theorem 3]{Mukhin} for the 2-dimensional case.

% Fix $s\in\mathbb{N}$. For a metric space $(\mathbb{R}^s,d)$ with a metric $d$,
% we denote by $B^d_r$ the closed ball of radius $r>0$ centered at $0$.
% For $x=(x_1,\ldots,x_s),y=(y_1,\ldots,y_s)\in\mathbb{R}^s$, we denote by $\langle x,y\rangle$
% the scalar product of $x,y$, i.e. $\langle x,y\rangle=x_1y_1+\ldots+x_sy_s$.

Let $\langle X,Y\rangle$ denote the scalar product of two vectors $X,Y\in\mathbb{R}^2$.
Let $\|X\|_2$ denote the Euclidean norm and $\|X\|_\infty$ denote the maximum norm in $\mathbb{R}^2$.
For a random vector $X$, let $X^*$ be the symmetrized version of $X$,
that is, $X^*=X+Z$ where $Z$ is an independent copy of $-X$.

For a positive integer $N$, let $X_{N,1},\ldots,X_{N,N}$ be independent random vectors in $\mathbb{Z}^2$.
Set $S_N=\sum_{k=1}^N X_{N,k}$. For $u>0$, define 
$$
\tau(u)=\inf_{\|T\|_2=1}\sum_{k=1}^N\mean\of{\langle X_{N,k}^*,T\rangle^2\mathbb{I}(\|X_{N,k}^*\|_2\leq u)},
$$
where $\mathbb{I}(\mathcal{A})$ denotes the indicator random variable of an event $\mathcal{A}$.
For a real number $x$, let $\{x\}$ denote the minimum distance from $x$ to an integer.
For $T\in\mathbb{R}^2$, we set 
$$
H_N(T)=\sum_{k=1}^N\mean\{\langle X^*_{N,k},T\rangle\}^2.
$$

\begin{proposition}[An ILT-to-LLT conversion (Mukhin~\cite{Mukhin})]
Let $\sigma_N>0$ be such that $\sigma_N\to\infty$ as $N\to\infty$.
Assume that $\sigma_N^{-1}S_N\stackrel{d}\to S$, where $S$ is an absolutely continuous random vector with density $s$.
Assume also that there exists a sequence of reals $m_N\in(4,\sigma_N)$ such that
\begin{enumerate}[(a)]
\item
% there exists $\alpha>0$ such that, for every $n\in\mathbb{N}$, $B_n^2(m_n)\geq\alpha \sigma_n^2$;
$\tau(m_N)=\Omega(\sigma_N^2)$ and
\item
$m_N=o\of{\exp\of{\inf_{\|T\|_\infty\le1/2,\,\|T\|_2>1/m_N}H_N(T)}}$.
\end{enumerate}
%  t\in T(1/2,1/m_n)  where $T(a,b)=\Set{t}{\|t\|_\infty\le a,\,\|t\|_2>b}$.
%B^{\ell_{\infty}}_{1/2}\setminus B^{\ell_2}_{1/m_n}
Then
\begin{equation}
  \label{eq:Mukhin}
\prob{S_N=x}=\sigma_N^{-2} s(\sigma_N^{-1}x)+o(\sigma_N^{-2})  
\end{equation}
uniformly over all $x\in\mathbb{Z}^2$.
\label{th:local_CLT}
\end{proposition}

We now apply Proposition \ref{th:local_CLT} to the random vector
%$\Psi_1$, $\Psi_2$ and
$S_N=\Psi$ decomposed into the sum of $N={n\choose2}$ independent random vectors $X_{N,k}$
according to Equality \eqref{eq:Psi-decomp}. More precisely, we use the equality $\Psi=\sum_{e}(\psi_e,\psi'_e)$,
which is equivalent to \eqref{eq:Psi-decomp}. Let $S=(\eta_1,\eta_2)$ be the standard normal
vector in $\mathbb{R}^2$ and, correspondingly,
% $$
%  g_1(x)=g_2(x)=\frac{1}{\sqrt{2\pi}}e^{-x^2/2},\quad
%  g(x_1,x_2)=\frac{1}{2\pi\sqrt{c_1c_2}}e^{-x_1^2/(2c_1)-x_2^2/(2c_2)}
%  $$
$s(x_1,x_2)=\frac{1}{2\pi}e^{-x_1^2/2-x_2^2/2}$ be the density function of
the standard bivariate normal distribution. Due to \eqref{eq:integral_limit_th},
the convergence $\sigma_N^{-1}S_N\stackrel{d}\to S$ holds with the rescaling factor
% $$
%  \mu_{1;N}=0,\,\,\sigma_{1;N}=\sqrt{c_1}n,\quad 
%  \mu_{2;N}=0,\,\,\sigma_{2;N}=\sqrt{c_2}n,\quad 
%  \mu_N=(\mu_{1;N},\mu_{2;N}),\,\,\sigma_N=n.
%  $$
$\sigma_N=2^{-3/2}n$.

We need to show that Conditions (a) and (b) are fulfilled for some sequence of reals $m_N$.
Set $m_N=5$ for all $N$.
Suppose that $X_{N,k}=(\psi_e,\psi'_e)$ for some $e\in\bigcup_{a=1}^4S_a\cup\bigcup_{a=1}^4S'_a$.
%$e\in S_1\cup S_2\cup S_3\cup S_4\cup S'_1\cup S'_2\cup S'_3\cup S'_4$.
This random vector takes on the value $(0,0)$
with probability $1/2$ and an exactly one non-zero value $(v_1,v_2)\in\{-2,-1,0,-1,-2\}^2$
with the same probability. It follows that $X_{N,k}^*=(0,0)$ still with probability $1/2$
and $X_{N,k}^*=(v_1,v_2)$ or $X_{N,k}^*=(-v_1,-v_2)$ with probability $1/4$ in each case.
This implies that $\tau(m_N)=\inf_{\|T\|_2=1}\sum_{k=1}^N\mean\langle X_{N,k}^*,T\rangle^2$.
To estimate this value, consider $T=(t_1,t_2)$ with $\|T\|_2=1$. Note that
$\mean\langle X_{N,k}^*,T\rangle^2$ is equal to $\frac12(t_1+t_2)^2$ if $e\in S_1\cap S'_1$
and to $\frac12(2t_1+t_2)^2$ if $e\in S_3\cap S'_1$. It is easy to show\footnote{%
  If $t_2=0$, then $t_1=\pm1$, and the claim is straightforward. Otherwise, let $t_1=\alpha t_2$
  for a real $\alpha$. The condition $\|T\|_2=1$ implies that $t_2^2(1+\alpha^2)=1$ and, hence,
  $\frac12(t_1+t_2)^2+\frac12(2t_1+t_2)^2=f(\alpha)$ where
  $f(\alpha)=((1+\alpha)^2+(1+2\alpha)^2)/(2+2\alpha^2)$. The function $f$ attains its global
  minimum $(7-3\sqrt5)/4>1/16$ at $\alpha=(1-\sqrt5)/2$.}
that the sum of these two values is larger than $1/16$. It follows that
$$
\tau(m_N)>\frac1{16}\min\of{|S_1\cap S'_1|,|S_3\cap S'_1|}=\of{2^{-11}+o(1)}n^2,
$$
yielding Condition~(a).

In order to verify Condition (b), consider $T=(t_1,t_2)$ such that $\|T\|_\infty\le1/2$ and $\|T\|_2>1/5$.
It suffices to show that $H_N(T)=\Omega(N)$ with the constant factor hidden in the $\Omega$-notation
not depending on the particular choice of $T$. Note that $\langle X^*_{N,k},T\rangle=\pm(t_1+t_2)$
with probability $1/2$
if $X_{N,k}=(\psi_e,\psi'_e)$ for $e\in S_1\cap S'_1$ and $\langle X^*_{N,k},T\rangle=\pm(2t_1+t_2)$
with the same probability if $e\in S_3\cap S'_1$. It is not hard to see\footnote{%
  Indeed, $|t_1+t_2|\le|t_1|+|t_2|\le|t_1|+1/2$. On the other hand, if $|t_1|\le1/\sqrt5$,
  then $|t_1+t_2|\ge|t_2|-|t_1|>(1/5-t_1^2)^{1/2}-|t_1|$. Since the function $f(t)=(1/5-t^2)^{1/2}-t$
  is continuous at $t=0$, there is $\delta\in(0,1/4)$ such that $1/4<|t_1+t_2|<1/2+\delta$
  whenever $|t_1|\le\delta$. Therefore $|\{t_1+t_2\}|>1/4$ if $|t_1|\le\delta$.
  If $|t_1|>\delta$, then $|t_1|\in(\delta,1/2]$, and $|\{t_1+t_2\}|+|\{2t_1+t_2\}|>\delta$
  because $(2t_1+t_2)-(t_1+t_2)=t_1$.}
that $|\{t_1+t_2\}|+|\{2t_1+t_2\}|\ge\delta$ for a constant $\delta>0$ independent of $T$.
It follows that
$$
H_N(T)\ge\frac{\delta^2}4\min\of{|S_1\cap S'_1|,|S_3\cap S'_1|}=\Omega(n^2),
$$
as needed.

% Let us first verify the conditions for $\Phi_1$ (due to symmetry, same arguments work for $\Phi_2$ as well). Set $m_N=5$. Clearly, $B^2_N(m_N)=\sigma_N^2$. Therefore, the condition A holds true for $\alpha=1$. We then note that, for every $k$, the following five options are possible: $X_{N,k}\in\{-3,-2,-1,-0,1,2,3\}$;  $X_{N,k}\in\{-3,-2,-1,-0,1\}$;  $X_{N,k}\in\{-1,-0,1,2,3\}$; $X_{N,k}\in\{-1,0,1\}$ and $X_{N,k}\in\{0,1\}$. In each of the cases, every value has positive constant probability.  Thus, for every $t\in B^{\ell_{\infty}}_{1/2}\setminus B^{\ell_2}_{1/5}$, $\mathbb{E}\{tX^*_{N,k}\}^2$ is bounded away from $0$ implying that $H_N(t)=\Theta(n^2)$. We get $B$ immediately. It remains to verify both conditions for $\Phi$. As above, we set $m_N=5$. We get 
% $$
% B^2_N(m_N)=
% \inf_{\|t\|_2=1}\sum_{k=1}^N\mathbb{E}\langle X_{N,k}^*,t\rangle^2=\Theta(n^2),
% $$
% A follows. Since, for every $t\in B^{\ell_{\infty}}_{1/2}\setminus B^{\ell_2}_{1/5}$, $\mathbb{E}\langle X_{N,k}^*,t\rangle^2$ is bounded away from $0$, we get $B$ as well. 

Therefore, Proposition~\ref{th:local_CLT} is applicable. For $x=(0,0)$, it gives us
$$
\prob{\Psi_1=0,\,\Psi_2=0}=\frac{4}{\pi n^2}+o\of{\frac{1}{n^2}},
% \mathbb{P}(\Psi_1=0,\Psi_2=0)=\frac{1}{2\pi\sqrt{c_1c_2}n^2}%\exp\left(-\frac{(\mathbb{E}\Psi_1)^2}{2 c_1n^2}-\frac{(\mathbb{E}\Psi_2)^2}{2 c_2n^2}\right)
%  +o\left(\frac{1}{n^2}\right),
$$
% $$ 
%  \mathbb{P}(\Psi_1=0)=\frac{1}{\sqrt{2\pi c_1}n}%\exp\left(-\frac{(\mathbb{E}\Psi_1)^2}{2 c_1n^2}\right)
%  +o\left(\frac{1}{n}\right),\quad
%  \mathbb{P}(\Psi_2=0)=\frac{1}{\sqrt{2\pi c_2}n}%\exp\left(-\frac{(\mathbb{E}\Psi_2)^2}{2 c_2n^2}\right)
%  +o\left(\frac{1}{n}\right),
% $$
% and then %Therefore, due to \eqref{eq:d_2_limit},
% $$
%  \mathbb{P}(\Psi_1=0,\Psi_2=0)=\mathbb{P}(\Psi_1=0)\mathbb{P}(\Psi_2=0)+o\left(\frac{1}{n^2}\right).
% $$
completing the proof of Equality \eqref{eq:cond-ind}.
Note that the infinitesimal hidden in the additive term $o(n^{-2})$ can be chosen the same
for all standard quadruples $U,V,U',V'$ due to the uniform bounds for the corresponding
term $o(\sigma_N^{-2})$ in \eqref{eq:Mukhin} stated in \cite[Theorem 5]{Mukhin}.
The proof of Claim \ref{cl:0} is therewith complete. 
\end{subproofof}

% %to run bibtex:
% \bibliographystyle{abbrv}
% \bibliography{wm-vs-cr}

% \end{document}

\section{Comparing WM and CR}\label{s:WM-vs-CR}

\subsection{Color refinement}\label{ss:CR}

We begin with a formal description of the \emph{color refinement} algorithm (\emph{CR} for short).
CR operates on vertex-colored graphs but applies also to uncolored graphs
by assuming that their vertices are colored uniformly.
An input to the algorithm consists either of a single graph or a pair of graphs.
Consider the former case first.
For an input graph $G$ with initial coloring $C_0$, CR iteratively computes new colorings
\begin{equation}
  \label{eq:refinement}
C_{i}(x)=\of{C_{i-1}(x),\Mset{C_{i-1}(y)}_{y\in N(x)}},  
\end{equation}
where $\Mset{}$ denotes a multiset and $N(x)$ is the neighborhood of a vertex $x$.
Denote the partition of $V(G)$ into the color classes of $C_i$ by $\cP_i$.
Note that each subsequent partition $\cP_{i+1}$ is either finer than or equal to $\cP_i$.
If $\cP_{i+1}=\cP_i$, then $\cP_{j}=\cP_i$ for all $j\ge i$.
Suppose that the color partition stabilizes in the $t$-th round,
that is, $t$ is the minimum number such that $\cP_t=\cP_{t-1}$.
CR terminates at this point and outputs the coloring $C=C_t$.
Note that if the colors are computed exactly as defined by \refeq{refinement},
they will require exponentially long color names. To prevent this,
the algorithm renames the colors after each refinement step, using the same set
of no more than $n$ color names.
We say that a graph $G$ is \emph{CR-discrete} if $C(x)\ne C(x')$ for all $x\ne x'$.

If an input consists of two graphs $G$ and $H$, then it is convenient to
assume that their vertex sets $V(G)$ and $V(H)$ are disjoint.
In this case, CR is run on the vertex-disjoint union of $G$ and $H$ as described above.
The vertex colorings of $G$ and $H$ define an initial coloring $C_0$ of
the union $V(G)\cup V(H)$, which is iteratively refined according to \refeq{refinement}.
The color partition $\cP_i$ is defined exactly as above but now on the whole
set $V(G)\cup V(H)$. As soon as the color partition of $V(G)\cup V(H)$ stabilizes\footnote{%
Note that the stabilization on each of the sets $V(G)$ and $V(H)$ can occur earlier than on~$V(G)\cup V(H)$.},
CR terminates and outputs the current coloring $C=C_t$ of $V(G)\cup V(H)$.

We say that CR \emph{distinguishes} $G$ and $H$ if $\Mset{C(x)}_{x\in V(G)}\ne\Mset{C(x)}_{x\in V(H)}$.
A graph $G$ is called \emph{CR-identifiable} if it is distinguishable by CR from every non-isomorphic~$H$.

Note that every CR-discrete graph $G$ is CR-identifiable.
Indeed, suppose that $\Mset{C(x)}_{x\in V(G)}=\Mset{C(x)}_{x\in V(H)}$.
This means that $H$ is CR-discrete as well and that there is a bijection
$f\function{V(G)}{V(H)}$ such that $C(f(x))=C(x)$ for all $x\in V(G)$.
If there are vertices $u$ and $v$ in $G$ such that
$u$ and $v$ are adjacent but $f(u)$ and $f(v)$ are not or vice versa,
then the partition of $V(G)\cup V(H)$ into the color classes of $C$
is still unstable because, by the refinement rule~\eqref{eq:refinement},
the vertices $u$ and $f(u)$ have to receive distinct colors in the next round.
This contradiction shows that the mapping $f$ must be an isomorphism from $G$ to~$H$.

\subsection{Proof of Theorem \ref{thm:WM-vs-CR}: Parts 1 and 2}

Parts 1 and 2 of Theorem \ref{thm:WM-vs-CR} follow immediately from the lemma below.
We prove this lemma by a direct combinatorial argument. Alternatively, one can
use an algebraic approach in \cite[Theorem 2]{PowersS82} or the connection to
finite variable logics exploited in \cite[Lemma~4]{Dvorak10}.

\begin{lemma}\label{lem:CR}
  Let $G$ and $H$ be uncolored $n$-vertex graphs (the case $G=H$ is not excluded).
  Let $x\in V(G)$, $x'\in V(H)$, and $k$ be an arbitrary non-negative integer.
  Then $C_k(x)\ne C_k(x')$ whenever $w^G_{k}(x)\ne w^H_{k}(x')$.
\end{lemma}

\begin{proof}
  Using the induction on $k$, we prove that $w^G_{k}(x)=w^H_{k}(x')$ whenever $C_k(x)=C_k(x')$.
  In the base case of $k=0$, these equalities are equivalent just because they are both true
  by definition (recall that $w^G_0(x)=0$). Assume that $C_k(y)=C_k(y')$ implies $w^G_{k}(y)=w^H_{k}(y')$
  for all $y\in V(G)$ and $y'\in V(H)$. Let $C_{k+1}(x)=C_{k+1}(x')$. By the definition \eqref{eq:refinement} of the refinement step,
  we have $\Mset{C_k(y)}_{y\in N(x)}=\Mset{C_k(y)}_{y\in N(x')}$.
  Using the induction assumption, from here we derive the equality
  $\Mset{w^G_{k}(y)}_{y\in N(x)}=\Mset{w^H_{k}(y)}_{y\in N(x')}$.
  The equality $w^G_{k+1}(x)=w^H_{k+1}(x')$ now follows by noting that
  $w^G_{k+1}(x)=\sum_{y\in N(x)}w^G_{k}(y)$.
\end{proof}

\subsection{Proof of Theorem \ref{thm:WM-vs-CR}: Part 3}

We now construct a graph $G$ with the three desired properties (a)--(c).
Note that this graph can be used to produce infinitely many
examples separating the strength of WM and CR. The simplest way to do so
is just to add a non-isomorphic CR-discrete graph as another connected component.
An infinite supply of CR-discrete graphs is ensured by the result of
Babai, Erd\H{o}s, and Selkow \cite{BabaiES80} (as well as by Part 1 of Theorem \ref{thm:3paths})
or, alternatively, one can take an arbitrary tree with no non-trivial automorphism.

Let $\bZ_n$ denote the cyclic group with elements $0,1,\ldots,n-1$
and operation being the addition modulo $n$. Our construction is based
on the well-known Shrikhande graph; see, e.g., \cite{Sane15}.
This is the Cayley graph of the group $\bZ_4\times \bZ_4$
with connection set $\{\pm(1,0),\pm(0,1),\pm(1,1)\}$.
A natural drawing of the Shrikhande graph on the torus can be seen in both parts of Fig.~\ref{fig:AB}.

Recall that a graph $G$ is \emph{strongly regular} with parameters
$(n,d,\lambda,\mu)$ if it has $n$ vertices, every vertex in $G$ has $d$ neighbors
(i.e., $G$ is \emph{regular} of degree $d$), 
every two adjacent vertices of $G$ have $\lambda$ common neighbors, 
and every two non-adjacent vertices have $\mu$ common neighbors.
Denote the Shrikhande graph by $S$.
The following properties of this graph are useful for our purposes.
\begin{enumerate}
\item[(P1)]
  $S$ is a strongly regular graph with parameters $(16,6,2,2)$.
\item[(P2)]
  $S$ is an arc-transitive graph, that is, every two ordered pairs of adjacent
  vertices in $S$ can be mapped onto each other by automorphisms of~$S$.
\item[(P3)]
  The pairs $u,v$ of non-adjacent vertices in the graph are split into two categories
  depending on whether the two common neighbors of $u$ and $v$ are adjacent or not.
  In the former case we call $u,v$ a \emph{$D$-pair} (as $u$, $v$, and their two neighbors
  induce a diamond graph). In the latter case we call $u,v$ a \emph{$Q$-pair}
  (as $u$, $v$, and their neighbors induce a quadrilateral graph~$C_4$).
\end{enumerate}

\begin{figure}
\raisebox{-20mm}{$A$}\hspace{-6mm}
\begin{tikzcd}[
  cells={nodes={shape=circle, draw=black}},
  tikz/.code=\tikzset{#1},
  tikz={
    v/.style={fill=none},
    a/.style={draw=none, fill=none},
    r/.style={fill=pred},
    b/.style={fill=blue},
    g/.style={fill=pgreen},  
  },
  arrows=-, 
  ]
%top
|[a]| & |[a]|  & |[a]|  & |[a]|  & |[a]|  & |[a]|  \\
%3
|[a]|\rar[-latex, shorten <=6mm] & |[v]|\rar\uar[-latex, shorten >=4mm]\urar[-stealth, shorten >=9mm]  & |[v]|\rar\uar[-latex, shorten >=4mm]\urar[-{stealth}{stealth}, shorten >=9mm]  & |[v]|\rar\uar[-latex, shorten >=4mm]\urar[-{stealth}{stealth}{stealth}, shorten >=9mm]  & |[v]|\uar[-latex, shorten >=4mm]\rar[-latex, shorten >=6mm]\urar[-latex, shorten >=9mm]  & |[a]|  \\
%2
|[a]|\rar[-latex, shorten <=6mm]\urar[-stealth, shorten <=9mm] & |[v]|\rar\uar\urar  & |[v]|\rar\uar\urar  & |[r]|\rar\uar\urar\ular[phantom, "a_1", near start]  & |[v]|\uar\rar[-latex, shorten >=6mm]\urar[-stealth, shorten >=9mm]  & |[a]|  \\
%1
|[a]|\rar[-latex, shorten <=6mm]\urar[-{stealth}{stealth}, shorten <=9mm] & |[v]|\rar\uar\urar  & |[b]|\rar\uar\urar\ular[phantom, "a_2", near start]  & |[v]|\rar\uar\urar  & |[g]|\uar\rar[-latex, shorten >=6mm]\urar[-{stealth}{stealth}, shorten >=9mm]\drar[phantom, "a_3", near start]  & |[a]|  \\
%0
|[a]|\rar[-latex, shorten <=6mm]\urar[-{stealth}{stealth}{stealth}, shorten <=9mm] & |[v]|\rar\uar\urar  & |[v]|\rar\uar\urar  & |[v]|\rar\uar\urar  & |[v]|\uar\rar[-latex, shorten >=6mm]\urar[-{stealth}{stealth}{stealth}, shorten >=9mm]  & |[a]|  \\
%bottom
|[a]|\urar[-latex, shorten <=9mm] & |[a]|\uar[-latex, shorten <=4mm]\urar[-stealth, shorten <=9mm]  & |[a]|\uar[-latex, shorten <=4mm]\urar[-{stealth}{stealth}, shorten <=9mm]  & |[a]|\uar[-latex, shorten <=4mm]\urar[-{stealth}{stealth}{stealth}, shorten <=9mm]  & |[a]|\uar[-latex, shorten <=4mm]  & |[a]|   
\end{tikzcd}
\quad\raisebox{-20mm}{$B$}\hspace{-6mm}
\begin{tikzcd}[
  cells={nodes={shape=circle, draw=black}},
  tikz/.code=\tikzset{#1},
  tikz={
    v/.style={fill=none},
    a/.style={draw=none, fill=none},
    r/.style={fill=pred},
    b/.style={fill=blue},
    g/.style={fill=pgreen},  
  },
  arrows=-, 
  ]
%top
|[a]| & |[a]|  & |[a]|  & |[a]|  & |[a]|  & |[a]|  \\
%3
|[a]|\rar[-latex, shorten <=6mm] & |[v]|\rar\uar[-latex, shorten >=4mm]\urar[-stealth, shorten >=9mm]  & |[v]|\rar\uar[-latex, shorten >=4mm]\urar[-{stealth}{stealth}, shorten >=9mm]  & |[v]|\rar\uar[-latex, shorten >=4mm]\urar[-{stealth}{stealth}{stealth}, shorten >=9mm]  & |[v]|\uar[-latex, shorten >=4mm]\rar[-latex, shorten >=6mm]\urar[-latex, shorten >=9mm]  & |[a]|  \\
%2
|[a]|\rar[-latex, shorten <=6mm]\urar[-stealth, shorten <=9mm] & |[v]|\rar\uar\urar  & |[v]|\rar\uar\urar  & |[b]|\rar\uar\urar\ular[phantom, "b_2", near start]  & |[v]|\uar\rar[-latex, shorten >=6mm]\urar[-stealth, shorten >=9mm]  & |[a]|  \\
%1
|[a]|\rar[-latex, shorten <=6mm]\urar[-{stealth}{stealth}, shorten <=9mm] & |[v]|\rar\uar\urar  & |[r]|\rar\uar\urar\ular[phantom, "b_1", near start]  & |[v]|\rar\uar\urar  & |[g]|\uar\rar[-latex, shorten >=6mm]\urar[-{stealth}{stealth}, shorten >=9mm]\drar[phantom, "b_3", near start]  & |[a]|  \\
%0
|[a]|\rar[-latex, shorten <=6mm]\urar[-{stealth}{stealth}{stealth}, shorten <=9mm] & |[v]|\rar\uar\urar  & |[v]|\rar\uar\urar  & |[v]|\rar\uar\urar  & |[v]|\uar\rar[-latex, shorten >=6mm]\urar[-{stealth}{stealth}{stealth}, shorten >=9mm]  & |[a]|  \\
%bottom
|[a]|\urar[-latex, shorten <=9mm] & |[a]|\uar[-latex, shorten <=4mm]\urar[-stealth, shorten <=9mm]  & |[a]|\uar[-latex, shorten <=4mm]\urar[-{stealth}{stealth}, shorten <=9mm]  & |[a]|\uar[-latex, shorten <=4mm]\urar[-{stealth}{stealth}{stealth}, shorten <=9mm]  & |[a]|\uar[-latex, shorten <=4mm]  & |[a]|   
\end{tikzcd}
\caption{Two colored versions of the Shrikhande graph.
  Unexposed edges are obtainable by identification of the arrows according to
the standard square representation of a torus.}
\label{fig:AB}
\end{figure}
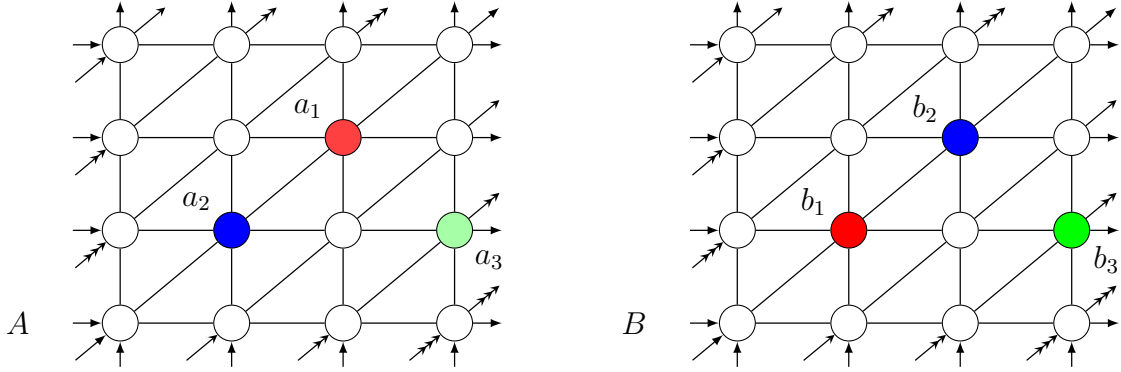

We now define two colored copies $A$ and $B$ of the Shrikhande graph.
Three vertices of $A$, denoted by $a_1$, $a_2$, and $a_3$, are
individualized by coloring them in red, blue, and green respectively
in such a way that
\begin{enumerate}
\item[(A1)]
  $a_1$ and $a_2$ are adjacent,
\item[(A2)]
  $a_1$ and $a_3$ are non-adjacent and form a $Q$-pair, and
\item[(A3)]
  $a_2$ and $a_3$ are non-adjacent and form a $D$-pair.
\end{enumerate}
The remaining, non-individualized vertices are considered also colored, all in the same color.
The graph $B$ is obtained from $A$ by interchanging the colors of $a_1$ and $a_2$; see Fig.~\ref{fig:AB}.
We denote the individualized vertices in $B$ by $b_1$, $b_2$, and $b_3$ so that
$a_i$ and $b_i$ are equally colored for each $i=1,2,3$.

We remark that $A$ is defined by Conditions (A1)--(A3) uniquely up to isomorphism
of vertex-colored graphs. Indeed, the adjacent pair $a_1,a_2$ can be chosen
arbitrarily due to Property (P2). Assuming a natural coordinatization for Fig.~\ref{fig:AB}
where the vertex in the left bottom corner is $(0,0)\in\bZ_4\times \bZ_4$,
we have $a_1=(2,2)$ and $a_2=(1,1)$. Besides $a_3=(3,1)$, there is only one more
vertex $a'_3=(1,3)$ satisfying Conditions (A2)--(A3). The choice of $a'_3$
instead of $a_3$ yields, however, a graph isomorphic to $A$ due to the observation that
the transposition $(x,y)\mapsto(y,x)$ is an automorphism of the Shrikhande graph.

This argument also implies that $B$ is, up to isomorphism of vertex-colored graphs,
uniquely determined by the conditions that $b_1$ and $b_2$ are adjacent,
$b_1$ and $b_3$ form a $D$-pair, and $b_2$ and $b_3$ form a $Q$-pair.
The colored graphs $A$ and $B$ are clearly non-isomorphic.

Before presenting further details, we give a brief outline of the rest of the proof.
We will begin with establishing some useful properties of $A$ and $B$. 
Though these graphs are non-isomorphic, it is useful to notice that they
are still quite similar in the sense that they are indistinguishable by one
round of CR (Claim \ref{cl:1eq}). On the other hand, both $A$ and $B$ are
CR-discrete (Claim \ref{cl:discr}) and are, therefore, distinguished
after CR makes sufficiently many rounds (Claim \ref{cl:dist}).
The desired graph $G$ will be constructed from $A$ and $B$ by
connecting the equally colored vertices, i.e., $a_i$ and $b_i$,
via new edges and vertices. While $a_1,a_2,a_3,b_1,b_2,b_3$ are not
colored any more in $G$, their neighborhoods are modified so that
their colors are actually simulated by iterated degrees.
This allows us to derive from Claims \ref{cl:discr} and \ref{cl:dist}
that $G$ is CR-discrete (Claim \ref{cl:G-discr}). On the other hand,
$G$ is not WM-discrete (Claim \ref{cl:G-not-discr}). In order to show
that some vertices in $G$ have the same numbers of outgoing walks of each length,
we use basic properties of strongly regular graphs (Claim \ref{cl:reg})
and the fact that a walk can leave $A$ or $B$ only via one of the vertices $a_1,a_2,a_3,b_1,b_2,b_3$
(and here a crucial role is also played by Claim \ref{cl:1eq}). Finally,
we argue that $G$ is not WM-identifiable (Claim~\ref{cl:G-not-id}).
Indeed, if we construct another graph $G'$ similarly to $G$ but using another copy of $A$ in place of $B$,
then $G$ and $G'$ will have the same walk matrix.

We now proceed with the detailed proof.

\setcounter{claim}{0}

\begin{claim}\label{cl:1eq}
  After the first round of CR, the vertex-colored graphs $A$ and $B$ are still
  indistinguishable, that is, $\Mset{C_1(x)}_{x\in V(A)}=\Mset{C_1(x)}_{x\in V(B)}$.
%That is, there is a bijection $f\function{V(A)}{V(B)}$ such that $C_1(x)=C_1(f(x))$ for all $x\in V(A)$.
\end{claim}

\newsavebox{\vvv}
\savebox{\vvv}{\begin{tikzcd}[
  cells={nodes={shape=circle, draw=black, minimum size=3mm}},
  tikz/.code=\tikzset{#1},
  tikz={
    rb/.style={node split={90,210,330}, node split color 3=white},
  },
  arrows=-, 
  ]
  |[rb]|
\end{tikzcd}}
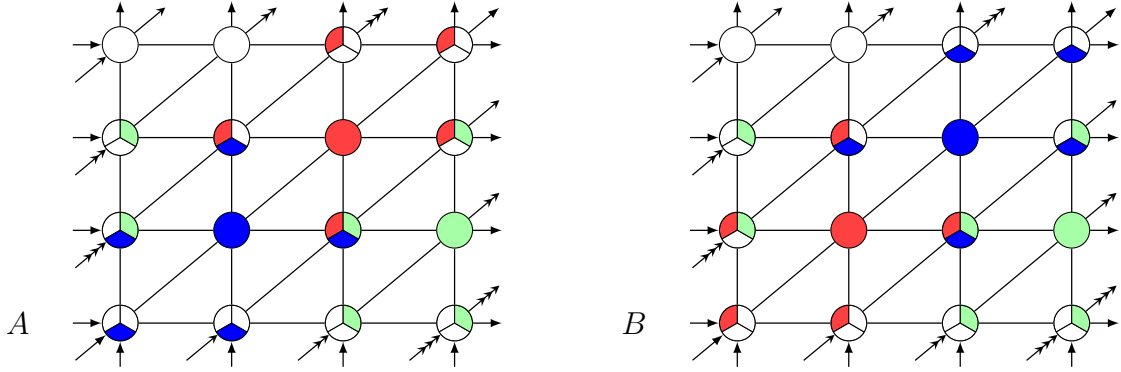
\begin{figure}
\raisebox{-20mm}{$A$}\hspace{-6mm}
\begin{tikzcd}[
  cells={nodes={shape=circle, draw=black}},
  tikz/.code=\tikzset{#1},
  tikz={
    v/.style={fill=none},
    a/.style={draw=none, fill=none},
    r/.style={fill=pred},
    b/.style={fill=blue},
    g/.style={fill=pgreen},
    rbg/.style={node split={90,210,330}},
    bg/.style={node split={90,210,330}, node split color 1=white},
    rg/.style={node split={90,210,330}, node split color 2=white},
    rb/.style={node split={90,210,330}, node split color 3=white},
    rr/.style={node split={90,210,330}, node split color 2=white, node split color 3=white},
    bb/.style={node split={90,210,330}, node split color 1=white, node split color 3=white},
    gg/.style={node split={90,210,330}, node split color 1=white, node split color 2=white},    
  },
  arrows=-, 
  ]
%top
|[a]| & |[a]|  & |[a]|  & |[a]|  & |[a]|  & |[a]|  \\
%3
|[a]|\rar[-latex, shorten <=6mm] & |[v]|\rar\uar[-latex, shorten >=4mm]\urar[-stealth, shorten >=9mm]  & |[v]|\rar\uar[-latex, shorten >=4mm]\urar[-{stealth}{stealth}, shorten >=9mm]  & |[rr]|\rar\uar[-latex, shorten >=4mm]\urar[-{stealth}{stealth}{stealth}, shorten >=9mm]  & |[rr]|\uar[-latex, shorten >=4mm]\rar[-latex, shorten >=6mm]\urar[-latex, shorten >=9mm]  & |[a]|  \\
%2
|[a]|\rar[-latex, shorten <=6mm]\urar[-stealth, shorten <=9mm] & |[gg]|\rar\uar\urar  & |[rb]|\rar\uar\urar  & |[r]|\rar\uar\urar  & |[rg]|\uar\rar[-latex, shorten >=6mm]\urar[-stealth, shorten >=9mm]  & |[a]|  \\
%1
|[a]|\rar[-latex, shorten <=6mm]\urar[-{stealth}{stealth}, shorten <=9mm] & |[bg]|\rar\uar\urar  & |[b]|\rar\uar\urar  & |[rbg]|\rar\uar\urar  & |[g]|\uar\rar[-latex, shorten >=6mm]\urar[-{stealth}{stealth}, shorten >=9mm]  & |[a]|  \\
%0
|[a]|\rar[-latex, shorten <=6mm]\urar[-{stealth}{stealth}{stealth}, shorten <=9mm] & |[bb]|\rar\uar\urar  & |[bb]|\rar\uar\urar  & |[gg]|\rar\uar\urar  & |[gg]|\uar\rar[-latex, shorten >=6mm]\urar[-{stealth}{stealth}{stealth}, shorten >=9mm]  & |[a]|  \\
%bottom
|[a]|\urar[-latex, shorten <=9mm] & |[a]|\uar[-latex, shorten <=4mm]\urar[-stealth, shorten <=9mm]  & |[a]|\uar[-latex, shorten <=4mm]\urar[-{stealth}{stealth}, shorten <=9mm]  & |[a]|\uar[-latex, shorten <=4mm]\urar[-{stealth}{stealth}{stealth}, shorten <=9mm]  & |[a]|\uar[-latex, shorten <=4mm]  & |[a]|   
\end{tikzcd}
\quad\raisebox{-20mm}{$B$}\hspace{-6mm}
\begin{tikzcd}[
  cells={nodes={shape=circle, draw=black}},
  tikz/.code=\tikzset{#1},
  tikz={
    v/.style={fill=none},
    a/.style={draw=none, fill=none},
    r/.style={fill=pred},
    b/.style={fill=blue},
    g/.style={fill=pgreen},
    rbg/.style={node split={90,210,330}},
    bg/.style={node split={90,210,330}, node split color 1=white},
    rg/.style={node split={90,210,330}, node split color 2=white},
    rb/.style={node split={90,210,330}, node split color 3=white},
    rr/.style={node split={90,210,330}, node split color 2=white, node split color 3=white},
    bb/.style={node split={90,210,330}, node split color 1=white, node split color 3=white},
    gg/.style={node split={90,210,330}, node split color 1=white, node split color 2=white},       
  },
  arrows=-, 
  ]
%top
|[a]| & |[a]|  & |[a]|  & |[a]|  & |[a]|  & |[a]|  \\
%3
|[a]|\rar[-latex, shorten <=6mm] & |[v]|\rar\uar[-latex, shorten >=4mm]\urar[-stealth, shorten >=9mm]  & |[v]|\rar\uar[-latex, shorten >=4mm]\urar[-{stealth}{stealth}, shorten >=9mm]  & |[bb]|\rar\uar[-latex, shorten >=4mm]\urar[-{stealth}{stealth}{stealth}, shorten >=9mm]  & |[bb]|\uar[-latex, shorten >=4mm]\rar[-latex, shorten >=6mm]\urar[-latex, shorten >=9mm]  & |[a]|  \\
%2
|[a]|\rar[-latex, shorten <=6mm]\urar[-stealth, shorten <=9mm] & |[gg]|\rar\uar\urar  & |[rb]|\rar\uar\urar  & |[b]|\rar\uar\urar  & |[bg]|\uar\rar[-latex, shorten >=6mm]\urar[-stealth, shorten >=9mm]  & |[a]|  \\
%1
|[a]|\rar[-latex, shorten <=6mm]\urar[-{stealth}{stealth}, shorten <=9mm] & |[rg]|\rar\uar\urar  & |[r]|\rar\uar\urar  & |[rbg]|\rar\uar\urar  & |[g]|\uar\rar[-latex, shorten >=6mm]\urar[-{stealth}{stealth}, shorten >=9mm]  & |[a]|  \\
%0
|[a]|\rar[-latex, shorten <=6mm]\urar[-{stealth}{stealth}{stealth}, shorten <=9mm] & |[rr]|\rar\uar\urar  & |[rr]|\rar\uar\urar  & |[gg]|\rar\uar\urar  & |[gg]|\uar\rar[-latex, shorten >=6mm]\urar[-{stealth}{stealth}{stealth}, shorten >=9mm]  & |[a]|  \\
%bottom
|[a]|\urar[-latex, shorten <=9mm] & |[a]|\uar[-latex, shorten <=4mm]\urar[-stealth, shorten <=9mm]  & |[a]|\uar[-latex, shorten <=4mm]\urar[-{stealth}{stealth}, shorten <=9mm]  & |[a]|\uar[-latex, shorten <=4mm]\urar[-{stealth}{stealth}{stealth}, shorten <=9mm]  & |[a]|\uar[-latex, shorten <=4mm]  & |[a]|   
\end{tikzcd}
\caption{The colorings of $A$ and $B$ after the first color refinement round.
  For each $i=1,2,3$, the vertices $a_i$ and $b_i$ have the same unique color.
  The color of each non-individualized vertex is determined by its adjacency to the individualized vertices.
  For example, the color \raisebox{1mm}{\scalebox{.75}{\usebox{\vvv}}}
    of a vertex in $A$ means
  that this vertex is adjacent to $a_1$ and $a_2$ but not to~$a_3$.}
\label{fig:1eq}
\end{figure}

\begin{subproof}
  Since $A$ is a regular graph, the color $C_1(x)$ of $x\in V(A)$ is determined by $C_0(x)$
  and the individualized neighbors of $x$, i.e., $N(x)\cap\{a_1,a_2,a_3\}$.
  Moreover, if $C_0(x)=C_0(x')$, then we have also $C_1(x)=C_1(x')$ exactly when $N(x)\cap\{a_1,a_2,a_3\}=N(x')\cap\{a_1,a_2,a_3\}$.
  After interchanging the colors of $a_1$ and $a_2$, the following happens.
  \begin{itemize}
  \item
    The $C_1$-colors of $a_1$ and $a_2$ are also interchanged because neither $a_1$ nor $a_2$
    is adjacent to $a_3$.
  \item
    The $C_1$-color of a vertex $x$ stays the same if $x$ is adjacent to both $a_1$ and $a_2$
    or non-adjacent to both of them.
  \item
    We are left with 6 vertices that are split by the coloring $C_1$ in four monochromatic classes.
    After swapping red and blue, two of these classes, namely $(N(a_1)\setminus N(a_2))\cap N(a_3)$
    and $(N(a_2)\setminus N(a_1))\cap N(a_3)$, are swapped. The same is true for the other two classes
    $(N(a_1)\setminus N(a_2))\setminus N(a_3)$ and $(N(a_2)\setminus N(a_1))\setminus N(a_3)$.
  \end{itemize}
  Taking into account the equality $|(N(a_1)\setminus N(a_2))\cap N(a_3)|=|(N(a_2)\setminus N(a_1))\cap N(a_3)|$,
  we see that the recoloring does not change the frequency of any $C_1$-color.
  In other words, the multiset of colors $\Mset{C_1(x)}_x$ stays the same after converting $A$ into $B$.
The colorings of the two graphs after the first refinement round are shown in Fig.~\ref{fig:1eq}.
\end{subproof}

\begin{claim}\label{cl:discr}
  Both vertex-colored graphs $A$ and $B$ are CR-discrete.
\end{claim}

\begin{subproof}
  The claim can be proved directly by applying one more refinement step to the colorings in Fig.~\ref{fig:1eq}.
  However, we prefer to use a more general argument that will later be used once again
  in the proof of Claim \ref{cl:G-discr} in a slighly different situation.
  
  Call a vertex \emph{solitary} if CR colors it differently than the other vertices of the graphs.
  We prove that every vertex in $A$ is solitary. Virtually the same argument applies also
  to $B$. The individualized vertices $a_1,a_2,a_3$ are solitary from the very beginning.
  The single vertex $a$ adjacent to all of them is obviously also solitary. Thus, $A$ contains
  a triangle subgraph whose all vertices, namely $a,a_1,a_2$, are solitary.
  Let $a'$ be the common neighbor of $a_1$ and $a_2$ different from $a$
  (recall that the Shrikhande graph is strongly regular with the third parameter $\lambda=2$).
  The fact that $a_1$ and $a_2$ are solitary implies that the equality $C(a')=C(x)$
  for $x\ne a'$ can be true only if $x=a$, which is actually impossible because
  $a$ is solitary. Therefore, $a'$ is solitary too. This argument applies to any
  triangle whose all vertices are solitary and to the other common neighbor of any two
  vertices of this triangle. Consider the graph whose vertices are the triangles of the Shrikhande graph,
  adjacent exactly when they share an edge. This graph (known as the \emph{Dyck graph}) is obviously connected,
  which readily implies that all vertices of $A$ are solitary.
\end{subproof}

\begin{claim}\label{cl:dist}
  The vertex-colored graphs $A$ and $B$ are distinguishable by CR.
\end{claim}

\begin{subproof}
By Claim \ref{cl:discr}, both $A$ and $B$ are CR-discrete and, therefore, CR-identifiable.
Therefore, CR distinguishes $A$ and $B$ just because these graphs are non-isomorphic
(recall that the two common neighbors of $b_2$ and $b_3$
are adjacent while the two common neighbors of $a_2$ and $a_3$ are not).
\end{subproof}

\begin{figure}
\pgfdeclarelayer{background layer}
\pgfsetlayers{background layer,main}
\centering
\begin{tikzpicture}

    \matrix[column sep=2cm,row sep=1cm,every node/.style={circle,draw,inner sep=2pt,fill=none}] {
&\node (c4) {$c_4$};&\\
\node (a3) {$a_3$};&\node (c3) {$c_3$};&\node (b3) {$b_3$};\\
\node (a2) {$a_2$};&\node (c2) {$c_2$};&\node (b2) {$b_2$};\\
\node (a1) {$a_1$};&\node (c1) {$c_1$};&\node (b1) {$b_1$};\\
\node[draw=none] (A) {$A$};&&\node[draw=none] (B) {$B$};\\
    };

\draw  (a2) -- (a1) -- (c1) -- (b1) -- (b2) (a2) -- (c2) -- (b2) (a3) -- (c3) -- (b3) (c2) -- (c3) -- (c4);

\begin{pgfonlayer}{background layer}
  \draw[fill=gray!30,draw=none] (a2) ellipse (1.5cm and 3cm);
  \draw[fill=gray!30,draw=none] (b2) ellipse (1.5cm and 3cm);  
\end{pgfonlayer}
  
\end{tikzpicture}  
\caption{Construction of $G$.}
\label{fig:G}
\end{figure}
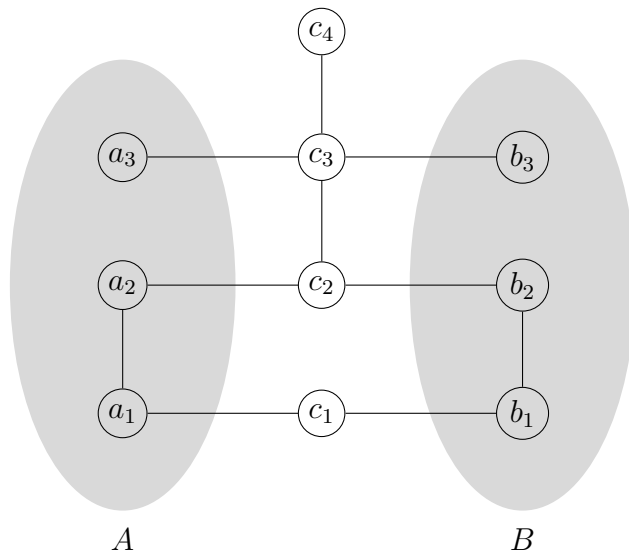

We now construct a graph $G$ as the vertex-disjoint union of $A$ and $B$
where each pair $a_i,b_i$ is connected via new edges and new intermediate vertices
as shown in Fig.~\ref{fig:G}.
Thus, $V(G)=V(A)\cup V(B)\cup\{c_1,c_2,c_3,c_4\}$ where $c_1,c_2,c_3,c_4$
are new connector vertices of degree $2,3,4,1$ respectively.
The graph $G$ is uncolored, that is, the colors of the six individualized vertices
$a_1,a_2,a_3,b_1,b_2,b_3$ are erased. The next claim proves Part 3(a) of the theorem.

\begin{claim}\label{cl:G-discr}
  $G$ is CR-discrete.
\end{claim}

\begin{subproof}
We use the terminology introduced in the proof of Claim \ref{cl:discr}.
  The connector vertices $c_1,c_2,c_3,c_4$ have unique degrees $2,3,4,1$
  and become solitary after the first refinement round.
  The vertices $a_1,a_2,a_3$ have degree 7, while the other vertices in $A$
  have degree 6. Each of the three vertices $a_1,a_2,a_3$ is distinguished from the other
  two by the adjacency to its own connector. It follows that after the second refinement
  round, the colors $C_2(a_1),C_2(a_2),C_2(a_3)$ become unique within $A$
  (even when still $C_2(a_i)=C_2(b_i)$). The argument used to prove Claim \ref{cl:discr}, therefore, implies
  that eventually $C_r(x)\ne C_r(x')$ for all $x\ne x'$ in $A$. The same argument applies to~$B$.

It remains to prove that $C(x)\ne C(x')$ if $x\in V(A)$ and $x'\in V(B)$.
Suppose that the vertices of $A$ receive pairwise distinct colors after the $r$-th round
and the same is true for $B$. If $A$ contains a vertex $x$ whose color $C_r(x)$
is absent in $B$, then the neighbors of $x$ in $A$ will be colored differently from
any vertex in $B$ in the next refinement round. Since $A$ has diameter 2,
one more round ensures that all vertices in $A$ are colored differently
from all vertices in $B$. There remains the case that there is a bijection
$f\function{V(A)}{V(B)}$ such that $C_r(f(x))=C_r(x)$ for all $x\in V(A)$.
Since $f$ is not an isomorphism from $A$ to $B$, the next refinement round
ensures the appearance of a vertex $x\in V(A)$ whose color $C_{r+1}(x)$
is absent in $B$ (similarly to the proof that the CR-discreteness implies CR-identifiability
in Subsection \ref{ss:CR}). Thus, we come back to the case already analyzed.
\end{subproof}

Let $w^R_k(x,y)$ denote the number of walks of length $k$ from a vertex $x$ to a vertex $y$ in a graph $R$.
We will need the following simple and well-known facts.\footnote{%
  Let $P_{k+1}$ be a path of length $k$ with end vertices $s$ and $t$.
  Note that $w^R_k(x,y)$ is equal to the number of all homomorphisms from $P_{k+1}$
  to $R$ taking $s$ to $x$ and $t$ to $y$. 
  Part 2 of Claim \ref{cl:reg} is a particular case of a much more general result
  about the invariance of homomorphism counts under the Weisfeiler-Leman equivalence
  for graphs with designated vertices \cite[Lemma~4]{Dvorak10}.
}

\begin{claim}\label{cl:reg}\hfill
  \begin{enumerate}
  \item
    If $R$ is a regular graph of degree $d$, then $w^R_k(x)=d^k$ for every $x\in V(R)$.
  \item
    Suppose now that $R$ is a strongly regular graph with parameters $(n,d,\lambda,\mu)$
    and fix an arbitrary $k\ge0$. Then the walk count $w^R_k(x,x)$ is the same for every $x\in V(R)$.
    If $x\ne y$, then the value of $w^R_k(x,y)$ depends only on the adjacency of $x$ and $y$
    (and on the parameters $d,\lambda,\mu$).
  \end{enumerate}
\end{claim}

\begin{subproof}
  Part 1 is obvious. Part 2 follows from an easy inductive argument. Indeed, it is trivially
  true for $k=0$. Assume that $w^R_k(x,y)=a_k$ for all adjacent $x$ and $y$,
  $w^R_k(x,y)=n_k$ for all non-adjacent unequal $x$ and $y$,
  and $w^R_k(x,x)=l_k$ for all $x$. Then
  $w^R_{k+1}(x,x)=\sum_{z\in N(x)}w^R_k(z,x)=d\,a_k$. If $x$ and $y$ are adjacent, then
  \begin{multline*}
  w^R_{k+1}(x,y)=\sum_{z\in N(x)\cap N(y)}w^R_k(z,y)+\sum_{z\in N(x)\setminus(N(y)\cup\{y\})}w^R_k(z,y)+w^R_k(y,y)\\
  =\lambda a_k+(d-\lambda-1)n_k+l_k.
  \end{multline*}  
  If $x$ and $y$ are non-adjacent and unequal, then
  $$
  w^R_{k+1}(x,y)=\sum_{z\in N(x)\cap N(y)}w^R_k(z,y)+\sum_{z\in N(x)\setminus N(y)}w^R_k(z,y)=\mu a_k+(d-\mu)n_k,
  $$
  enabling the induction step.
\end{subproof}

We are now prepared to prove Part 3(b) of the theorem.

\begin{claim}\label{cl:G-not-discr}
  $G$ is not WM-discrete.
\end{claim}

\begin{subproof}
  Define an equivalence relation $\equiv$ on $V(G)$ as follows.
  Each connector vertex is equivalent only to itself. Let $C_1$ be the coloring
  of $V(A)\cup V(B)$ obtained after the first round of the execution of CR
  on the vertex-colored graphs $A$ and $B$; see Claim \ref{cl:1eq}.
  We set $x\equiv x'$ for $x,x'\in V(A)\cup V(B)$ if $C_1(x)=C_1(x')$.
  Recall that the largest equivalence class of $\equiv$ consists of six vertices
  (three uncolored vertices in $A$ adjacent to $a_3$ but neither to $a_1$ nor to $a_2$
  and three uncolored vertices in $B$ adjacent to $b_3$ but neither to $b_1$ nor to $b_2$).
  To demonstrate that $G$ is not WM-discrete, we claim that $w^G_k(x)=w^G_k(x')$ for every $k$ whenever $x\equiv x'$.
  Indeed, if $x\in V(A)$, then
  \begin{equation}
    \label{eq:wGkx}
w^G_k(x)=w^A_k(x)+\sum_{i=1}^3\sum_{j=0}^{k-1}w^A_j(x,a_i)w^G_{k-j-1}(c_i).    
  \end{equation}
Here, we separately consider the walks of length $k$ inside $A$ and the walks of length $k$
leaving $A$. A walk can leave $A$ only after visiting one of the vertices $a_1,a_2,a_3$.
If such a walk leaves $A$ first after the $j$-th step via $a_i$, it arrives at the connector $c_i$
and, starting from it, makes the remaining $k-j-1$ steps. The similar equality holds for $x\in V(B)$.

It remains to notice that the right hand side of \refeq{wGkx} and its analog for $B$ yield the same value
for all $x$ in the same $\equiv$-class. Indeed, let $x\equiv x'$ and suppose that $x\in A$ and $x'\in B$
(the cases $x,x'\in A$ and $x,x'\in B$ are completely similar).
Then $w^A_k(x)=w^B_k(x')=6^k$ by Part 1 of Claim \ref{cl:reg}.
Finally, for each $j$ the equalities $w^A_j(x,a_i)=w^B_j(x,b_i)$ for $i=1,2,3$
follow from Part 2 of Claim \ref{cl:reg} by the definition of the relation $\equiv$
and the characterization of $C_1$ in the beginning of the proof of Claim~\ref{cl:1eq}.
\end{subproof}

It remains to prove Part 3(c) of the theorem.

\begin{claim}\label{cl:G-not-id}
  $G$ is not WM-identifiable.
\end{claim}

\begin{subproof}
  Construct $G'$ in the same way as $G$ but using a copy $A'$ of $A$ instead of $B$.
  The graphs $G$ and $G'$ are non-isomorphic, basically because $A$ and $B$ are non-isomorphic
  as colored graphs. In particular, $G'$ has an automorphism fixing the connector vertices
  and transposing $A$ and $A'$, whereas $G$ has no non-trivial automorphism by Claim \ref{cl:G-discr}.
  Fix a colored-graph isomorphism $f'$ from $A'$ to $A$.
  Let $f$ be a bijection from $V(A)$ to $V(B)$ such that $C(f(x))=C(x)$ for all $x\in V(A)$; see Claim \ref{cl:1eq}.
  Define a bijection $F$ from $V(G')=V(A)\cup V(A')\cup\{c_1,c_2,c_3,c_4\}$ onto
  $V(G)=V(A)\cup V(B)\cup\{c_1,c_2,c_3,c_4\}$ so that $F(c_i)=c_i$ for $i=1,2,3,4$,
  the restriction of $F$ to $V(A)$ coincides with $f$, and the restriction of $F$
  to $V(A')$ is the isomorphism $f'$. The proof of Claim \ref{cl:G-not-discr} applies
  to the graph $G'$ virtually without changes. In particular, the analog of Equality \refeq{wGkx} for $G'$
  allows us to show by a simple induction on $k$ that $w^{G'}_k(x)=w^G_k(f(x))$ for $x\in V(A)$ and
  $w^{G'}_k(x')=w^G_k(f'(x'))$ for $x'\in V(A')$, as well as that
  $w^{G'}_k(c_i)=w^G_k(c_i)$ for $i=1,2,3,4$. Thus, for every $x\in V(G')$ we have
  $w^{G'}_k(x)=w^G_k(F(x))$ for all $k$, implying that $G$ and $G'$ are WM-indistinguishable.
\end{subproof}

The proof of Theorem \ref{thm:WM-vs-CR} is complete.

\section*{Acknowledgement}

Oleg Verbitsky was supported by DFG grant KO 1053/8--2. 

%to run bibtex:
\bibliographystyle{abbrv}
\bibliography{wm-vs-cr-6}

% \begin{thebibliography}{10}
%
% \end{thebibliography}

\end{document}